\documentclass[12pt,onecolumn,draftclsnofoot]{IEEEtran} 

\usepackage{url}
\usepackage{tikz}
\usepackage{pgf}
\usepackage{amsthm,amssymb,amsmath}
\usepackage{cite,url}
\usepackage{graphicx,subfigure}
\usepackage{times,verbatim,color,stackrel,bm}

\newtheorem{theorem}{Theorem}

\newtheorem{definition}{Definition}
\newtheorem{subsec:coding}{subsec:coding}

\newtheorem{lemma}{Lemma}

\newtheorem{ppty}{Property}

\begin{document}

\title{The Feasibility of Scalable Video Streaming over Femtocell Networks}

\author{
\authorblockN{Donglin Hu \ \ and \ \ Shiwen Mao \\}
\authorblockA{Department of Electrical and Computer Engineering  \\
Auburn University, Auburn, AL, USA} 
}

\maketitle

\pagestyle{plain}\thispagestyle{plain}


\begin{abstract}
 In this paper, we consider femtocell CR networks, where femto base stations (FBS) are deployed to greatly improve network coverage and capacity. We investigate the problem of generic data multicast in femtocell networks. We reformulate the resulting MINLP problem into a simpler form, and derive upper and lower performance bounds. Then we consider three typical connection scenarios in the femtocell network, and develop optimal and near-optimal algorithms for the three scenarios. Second, we tackle the problem of streaming scalable videos in femtocell CR networks. A framework is developed to captures the key design issues and trade-offs with a stochastic programming problem formulation. In the case of a single FBS, we develop an optimum-achieving distributed algorithm, which is shown also optimal for the case of multiple non-interfering FBS's. In the case of interfering FBS's, we develop a greedy algorithm that can compute near-opitmal solutions, and prove a closed-form lower bound on its performance.
\end{abstract}

\section{Introduction}
Due to the use of open space as transmission medium, capacity of wireless networks are usually limited by interference.  When a mobile user moves away from the base station, a considerably larger transmit power is needed to overcome attenuation, while causing interference to other users and deteriorating network capacity.  To this end, femtocells provide an effective solution that brings network infrastructure closer to mobile users.  A femtocell is a small (e.g., residential) cellular network, with a {\em femto base station} (FBS) connected to the owner's broadband wireline network~\cite{Chandrasekhar08, Kim09, Guvenc10}. The FBS serves approved users when they are within the coverage.  Among the many benefits, femtocells are shown effective on improving network coverage and capacity~\cite{Chandrasekhar08}.  Due to reduced distance, transmit power can be greatly reduced, leading to prolonged battery life, improved signal-to-interference-plus-noise ratio (SINR), and better spatial reuse of spectrum. 

Femtocells have received significant interest from the wireless industry. Although highly promising, many important problems should be addressed to fully harvest their potential, 
such as interference mitigation, resource allocation, synchronization, and QoS provisioning~\cite{Chandrasekhar08, Kim09}.  It is also critical for the success of this technology to support important applications such as real-time video streaming in femtocell networks. 

In this paper, we first investigate the problem of data multicast in femtocell networks. It is not atypical that many users may request for the same content, as often observed in wireline networks. By allowing multiple users to share the same downlink multicast transmission, significant spectrum and power savings can be achieved. 

In particular, we adopt {\em superposition coding} (SC) and {\em successive interference cancellation} (SIC), two well-known PHY techniques, for data multicast in femtocell networks~\cite{Goldsmith06}.  With SC, a compound signal is transmitted, consisting of multiple signals (or, layers) from different senders or from the same sender. With SIC, a strong signal can be first decoded, by treating all other signals as noise. Then the decoder will reconstruct the signal from the decoded bits, and subtract the reconstructed signal from the compound signal. The next signal will be decoded from the residual, by treating the remaining signals as noise. And so forth. A special strength of the SC with SIC approach is that it enables simultaneous unicast transmissions (e.g., many-to-one or one-to-many). It has been shown that SC with SIC is more efficient than PHY techniques with orthogonal channels~\cite{Goldsmith06, Li09}. 

We adopt SC and SIC for the unique femtocell network environment, and investigate how to enable efficient data multicast from the femtocells to multiple users. We formulate a Mixed Integer Nonlinear Programming (MINLP) problem, which is NP-hard in general. The objective is to minimize the total BS power consumption. Then we reformulate the MINLP problem into a simpler form, and derive upper and lower performance bounds. We also derive a simple heuristic scheme that assigns users to the BS's with a greedy approach. Finally, we consider three typical connection scenarios in the femtocell network, and develop optimal and near-optimal algorithms for the three scenarios. The proposed algorithms have low computational complexity, and are shown to outperform the heuristic scheme with considerable gains. 

Then, we investigate the problem of video streaming in femtocell cognitive radio (CR) networks.  We consider a femtocell network consisting of a {\em macro base station} (MBS) and multiple FBS's. The femtocell network is co-located with a primary network with multiple licensed channels. This is a challenging problem due to the stringent QoS requirements of real-time videos and, on the other hand, the new dimensions of network dynamics (i.e., channel availability) and uncertainties (i.e., spectrum sensing and errors) found in CR networks.

We adopt Scalable Video Coding (SVC) in our system~\cite{Hu10JSAC,Hu10TW}. SVC encodes a video into multiple substreams, subsets of which can be decoded to provide different quality levels for the reconstructed video~\cite{Wien07}.  Such scalability is very useful for video streaming systems, especially in CR networks, to accommodate heterogeneous channel availabilities and dynamic network conditions. We consider H.264/SVC medium grain scalable (MGS) videos, 
since MGS can achieve better rate-distortion performance over 
Fine-Granularity-Scalability (FGS), although it only has Network Abstraction Layer (NAL) unit-based granularity~\cite{Wien07}.

The unique femtocell network architecture and the scalable video allow us to develop a framework that captures the key design issues and trade-offs, and to formulate a {\em stochastic programming} problem.  It has been shown that the deployment of femtocells has a significant impact on the network performance~\cite{Chandrasekhar08}.  In this paper, we examine three deployment scenarios.  In the case of a single FBS, we apply {\em dual decomposition} to develop a distributed algorithm that can compute the optimal solution.  In the case of multiple non-interfering FBS's, we show that the same distributed algorithm can be used to compute optimal solutions. In the case of multiple interfering FBS's, we develop a greedy algorithm that can compute near-optimal solutions, and prove a closed-form lower bound for its performance based on an {\em interference graph} model.  The proposed algorithms are evaluated with simulations, and are shown to outperform three alternative schemes with considerable gains. 

The remainder of this paper is organized as follows.  The related work is discussed in Section~\ref{sec:femto_work}.  We investigate the problem of data multicast over fenmtocell networks in Section~\ref{sec:femto_mcast_sic}. The problem of streaming multiple MGS videos in a femtocell CR network is discussed in Section~\ref{sec:femto_cr_video}. Section~\ref{sec:femto_conc} concludes this paper. 

\section{Background and Related Work}\label{sec:femto_work}
Femtocells have attracted considerable interest from both industry and academia. Technical and business challenges, requirements and some preliminary solutions to femtocell networks are discussed in~\cite{Chandrasekhar08}. 
Since FBS's are distributedly located and are able to spatially reuse the same channel, considerable research efforts were made on interference analysis and mitigation~\cite{Chandrasekhar09, Lee10}. 
A distributed utility based SINR adaptation scheme was presented in~\cite{Chandrasekhar09} to alleviate cross-tire interference at the macrocell from co-channel femtocells. Lee, Oh and Lee~\cite{Lee10} proposed a fractional frequency reuse scheme to mitigate inter-femtocell interference. 

Deploying femtocells by underlaying the macrocell has been proved to significantly improve indoor coverage and system capacity. However, interference mitigation in a two-tier heterogeneous network is a challenging problem. In~\cite{Chu11}, the interference from macrocell and femtocells was mitigated by a spatial channel separation scheme with codeword-to-channel mapping. In~\cite{Rangan10}, the rate distribution in the macrocell was improved by subband partitioning and modest gains were achieved by interference cancellation. In~\cite{Bharucha09}, the interference was controlled by denying the access of femtocell base stations to protect the transmission of nearby macro base station. A novel algorithmic framework was presented in~\cite{Madan10} for dynamic interference management to deliver QoS, fairness and high system efficiency in LTE-A femtocell networks. Requiring no modification of existing macrocells, 
CR was shown to achieve considerable performance improvement when applied to interference mitigation~\cite{Cheng11}. In~\cite{Kaimaletu11}, the orthogonal time-frequency blocks and transmission opportunities were allocated based on a safe/victim classification.

SIC has high potential of sending or receiving multiple signals concurrently, which improves the transmission efficiency~\cite{Hu11GC}. In~\cite{Li09}, the authors developed MAC and routing protocols that exploit SC and SIC to enable simultaneous unicast transmissions.  
Sen, et al. investigated the possible throughput gains with SIC from a MAC layer perspective~\cite{Sen10}. Power control for SIC was comprehensively investigated and widely applied to code division multiple access (CDMA) systems~\cite{Jean09, Park08, Benvenuto07, Agrawal05, Andrews03}.  Applying game theory, Jean and Jabbari proposed an uplink power control under SIC in direct sequence-CDMA networks~\cite{Jean09}. In~\cite{Park08}, the authors introduced an iterative two-stage SIC detection scheme for a multicode MIMO system and showed the proposed scheme significantly outperformed the equal power allocation scheme. A scheme on joint power control and receiver optimization of CDMA transceivers was presented in~\cite{Benvenuto07}. 
In~\cite{Agrawal05, Andrews03}, the impact of imperfect channel estimation and imperfect interference cancellation on the capacity of CDMA systems was examined.  

The problem of video over CR networks has only been studied in a few recent papers~\cite{Shiang08, Ding09, Hu12JSAC, Luo11}. In our prior work, we investigated the problem of scalable video over infrastructure-based CR networks~\cite{Hu10JSAC} and multi-hop CR networks~\cite{Hu10TW}. The preliminary results of video over femtocell CR networks were presented in~\cite{Hu11IDS}.

\section{Multicast in Femtocell Networks with Superposition Coding and Successive Interference Cancellation}\label{sec:femto_mcast_sic}
In this section, we formulate a Mixed Integer Nonlinear Programming (MINLP) problem of data multicast in femotcell networks, which is NP-hard in general. Then we reformulate the MINLP problem into a simpler form, and derive upper and lower performance bounds. We also derive a simple heuristic scheme that assigns users to the BS's with a greedy approach. Finally, we consider three typical connection scenarios in the femtocell network, and develop optimal and near-optimal algorithms for the three scenarios. The proposed algorithms have low computational complexity, and are shown to outperform the heuristic scheme with considerable gains.

\subsection{System Model and Problem Statement \label{sec:mod3}}

\subsubsection{System Model}

Consider a femtocell network with an MBS (indexed $0$) and $M$ FBS's (indexed from $1$ to $M$) deployed in the area. The $M$ FBS's are connected to the MBS and the Internet via broadband wireline connections. Furthermore, we assume a spectrum band that is divided into two parts, one is allocated to the MBS with bandwidth $B_0$ and the other is allocated to the $M$ FBS's. The bandwidth allocated to FBS $m$ is denoted by $B_m$. When there is no overlap between the coverages of two FBS's, they can spatially reuse the same spectrum.  Otherwise, the MBS allocates disjoint spectrum to the FBS's with overlapping coverages. We assumed the spectrum allocation is known a priori.  

There are $K$ mobile users in the femtocell network. Each user is equipped with one transceiver that can be tuned to one of the two available channels, i.e., connecting to a nearby FBS or to the MBS. The network is time slotted. We assume block-fading channels, where the channel condition is constant in each time slot~\cite{Goldsmith06}. We focus on a multicast scenario, where the MBS and FBS's multicast a data file to the $K$ users. The data file is divided into multiple packets with equal length and transmitted in sequence with the same modulation scheme. Once packet $l$ is successfully received and decoded at the user, it requests packet $(l+1)$ in the next time slot. 


We adopt SC and SIC to transmit these packets~\cite{Goldsmith06}, as illustrated in Fig.~\ref{fig:sic}. In each time slot $t$, the compound signal has $L$ {\em layers} (or, levels), denoted as $D_1(t)$, $\cdots$, $D_L(t)$. Each level $D_i(t)$, $i=1, \cdots, L$, is a packet requested by some of the users in time slot $t$. 
A user that has successfully decoded $D_i(t)$, for all $i=1$, $\cdots$, $l-1$, is able to subtract these signals from the received compound signal and then decodes $D_l(t)$, while the signals from $D_{l+1}(t)$ to $D_L(t)$ are treated as noise. 

\subsubsection{Problem Statement}

For the SC and SIC scheme to work, the transmit powers for the levels should be carefully determined, such that there is a sufficiently high SNR for the levels to be decodable. 
It is also important to control the transmit powers of the BS's to  reduce interference and leverage frequency reuse. The annual power bill is a large part of a mobile operator's costs~\cite{Ulf10}. Minimizing BS power consumption is important to reduce not only the operator's OPEX, but also the global CO$_2$ emission; an important step towards ``green'' communications.  

\begin{figure}[!t]
\centering
\includegraphics[width=4.5in]{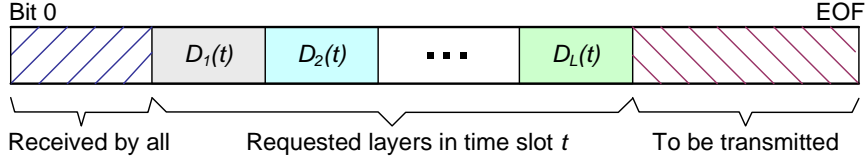}
\caption{Superposition coding and successive interference cancellation.}
\label{fig:sic}
\end{figure}

Therefore, we focus on BS power allocation in this paper. The objective is to minimize the total power of all the BS's, while guaranteeing a target rate $R_{tar}$ for each user. Recall that the data file is partitioned into equal-length packets. The target rate $R_{tar}$ ensures that a packet can be transmitted within a time slot, for given modulation and channel coding schemes.

Define binary indicator $I_m^k$, for all $m$ and $k$, as:
\begin{equation} \label{eq:imk}
	I_m^k = \left\{ \begin{array}{ll}
      1, & \mbox{if user $k$ connects to BS $m$} \\
      0, & \mbox{otherwise.}
                  \end{array} \right.
\end{equation}
Consider a general time slot $t$ when $L$ data packets, or levels, are requested. We formulate the optimal power allocation problem (termed OPT-Power) as follows. 
\begin{eqnarray}
\mbox{minimize:} && \hspace{-0.2in} \sum_{m=0}^M \sum_{l=1}^L P_l^m \label{eq:ObjFun11} \\
\mbox{subject to:} && \hspace{-0.2in} B_m\log_2(1+\gamma_m^k I_m^k) \ge R_{tar}I_m^k, \mbox{ for all } k \label{eq:cntrate} \\
 && \hspace{-0.2in} \sum_{m=0}^M I_m^k=1, \mbox{for all } k \label{eq:cnttransceiver} \\
 && \hspace{-0.2in} P_l^m \ge 0, \mbox{ for all } l, m,
\end{eqnarray}
where $P_l^m$ is the power of BS $m$ for transmitting the level $l$ packet; 
$\gamma_m^k$ is the SNR at user $k$ if it connects to BS $m$. Constraint (\ref{eq:cntrate}) guarantees the minimum rate at each user. Constraint (\ref{eq:cnttransceiver}) is due to the fact that each user is equipped with one transceiver, so it can only connect to one BS.  

Let $\mathcal{U}_l$ denote the set of users requesting the level $l$ packet. A user $k \in \mathcal{U}_l$ has decoded all the packets up to $D_{l-1}$. It subtracts the decoded signals from the received signal and treats signals $D_{l+1}, \cdots, D_L$ as noise. The SNR at user $k \in \mathcal{U}_l$, for $l=1, \cdots, L-1$, can be written as:
\begin{eqnarray}\label{eq:SNR1}
\gamma_m^k = H_m^k P_l^m / \left(N_0 + H_m^k \sum_{i=l+1}^L P_i^m \right),
\end{eqnarray}
where $H_m^k$ is the random channel gain from BS $m$ to user $k$ and $N_0$ is the noise power. For user $k \in \mathcal{U}_L$ that requests the last packet $D_L$, the SNR is
\begin{eqnarray}\label{eq:SNR2}
\gamma_m^k = H_m^k P_L^m / N_0.
\end{eqnarray}

The optimization variables in Problem OPT-Power consist of the binary variables $I_m^k$'s and the continuous variables $P_l^m$'s. It is an MINLP
problem, which is NP-hard in general. In Section~\ref{sec:alg}, we first reformulate the problem to a obtain a simpler form, and then develop effective algorithms for optimal and suboptimal solutions.  


\subsection{Reformulation and Power Allocation \label{sec:alg}}

In this section, we reformulate Problem OPT-Power to obtain a simpler form, and derive an upper bound and a lower bound for the total BS power. The reformulation also leads to a simple heuristic algorithm. Finally, we introduce power allocation algorithms for three connection scenarios. 

\subsubsection{Problem Reformulation}

Due to the monotonic logarithm functions and the binary indicators $I_m^k$, constraint~(\ref{eq:cntrate}) can be rewritten as:
%
\begin{equation}\label{eq:StSNR}
  \gamma_m^k I_m^k \ge \Gamma_m^k I_m^k, \;\; m=0, 1, \cdots,M,
\end{equation}
where $\Gamma_m^k = \Gamma_m =: 2^{R_{tar}/B_m} - 1$ is the minimum SNR requirement at user $k$ that connects to BS $m$. 
To further simplify the problem, define $Q_l^m = \sum_{i=l}^L P_i^m$, with $Q_{L+1}^m=0$. Then power $P_l^m$ is the difference 
\begin{eqnarray}\label{eq:QrepP}
  P_l^m=Q_l^m-Q_{l+1}^m.
\end{eqnarray}
Problem OPT-Power can be reformulated as:
\begin{eqnarray}
\mbox{minimize} && \hspace{-0.2in} \sum_{m=0}^M Q_1^m \label{eq:ObjFun2} \\
\mbox{subject to:} && \hspace{-0.2in} H_m^k(Q_l^m \hspace{-0.025in} - \hspace{-0.025in} Q_{l+1}^m) / \hspace{-0.025in} \left( N_0 \hspace{-0.025in} + \hspace{-0.025in} H_m^k Q_{l+1}^m \right) I_m^k \ge \Gamma_m I_m^k, 
 \nonumber \\ 
 && \hspace{0.7in} 
 \mbox{ for all } k \in \mathcal{U}_l, l=1, \cdots, L \label{eq:cntrate2} \\
 && \hspace{-0.2in} Q_l^m \ge Q_{l+1}^m, l=1, \cdots, L \\
 && \hspace{-0.2in} \sum_{m=0}^M I_m^k=1, \mbox{ for all } k. 
\end{eqnarray}

For $l\le L$, constraint~(\ref{eq:cntrate2}) can be rewritten as:
\begin{eqnarray}\label{eq:QmlIneq}
 Q_l^m I_m^k \ge \left[ N_0 \Gamma_m / H_m^k + (1 + \Gamma_m) Q_{l+1}^m \right] I_m^k.
\end{eqnarray}
Let $\mathcal{U}_l^m$ be the subset of users connecting to BS $m$ in $\mathcal{U}_l$. Since $Q_l^m \ge Q_{l+1}^m$, (\ref{eq:QmlIneq}) can be rewritten as,
\begin{equation}\label{eq:QmlEqu}
  Q_l^m = \max \left\{ Q_{l+1}^m, \max_{k \in \mathcal{U}_l^m} \left[ N_0 \Gamma_m / H_m^k + (1 + \Gamma_m) Q_{l+1}^m \right] \right\}. 
\end{equation}
From (\ref{eq:QmlEqu}), we define a function $Q_l^m = F_m(Q_{l+1}^m,\mathcal{U}_l^m)$ as:
\begin{eqnarray} \label{eq:FmDef}
F_m(Q_{l+1}^m,\mathcal{U}_l^m) 
= \left\{\begin{array}{l l}
 Q_{l+1}^m,& \mathcal{U}_l^m=\emptyset \\
 \max_{k\in\mathcal{U}_l^m} \left\{ \frac{N_0\Gamma_m}{H_m^k}+(1+\Gamma_m)Q_{l+1}^m \right\}, & \hspace{-0.05in} \mathcal{U}_l^m \neq \emptyset.
\end{array}\right.  
\end{eqnarray}
Obviously, $F_m(Q_{l+1}^m,\mathcal{U}_l^m)$ is non-decreasing with respect to $Q_{l+1}^m$.
%
It follows that
\begin{eqnarray}\label{eq:ObjFun3}
Q_1^m &=& F_m(Q_2^m, \mathcal{U}_1^m) \; = \; F_m(F_m(Q_3^m,\mathcal{U}_2^m),\mathcal{U}_1^m) \nonumber \\
&=& F_m(\cdots (F_m(Q_{L+1}^m, \mathcal{U}_L^m), \mathcal{U}_{L-1}^m), \cdots, \mathcal{U}_1^m) \nonumber \\
&=& F_m(\cdots (F_m(0, \mathcal{U}_L^m), \mathcal{U}_{L-1}^m), \cdots, \mathcal{U}_1^m).
\end{eqnarray}

If none of the subsets $\mathcal{U}_l^m$ ($l=1,\cdots,L$) is empty, we can expand the above recursive term using (\ref{eq:FmDef}). It follows that
%
\begin{equation}\label{eq:FoldTerm}
Q_1^m = N_0 \Gamma_m \sum_{l=1}^L (1 + \Gamma_m)^{c_l^m} \max_{k \in \mathcal{U}_l^m} \left\{1 / H_m^k \right\},
\end{equation}
where the exponent $c_l^m$ is defined as $c_1^m=0$ and $c_{l+1}^m=c_l^m+1$.  Otherwise, if a subset $\mathcal{U}_l^m = \emptyset$ for some $m$, we have that $Q_l^m=Q_{l+1}^m$,
$\max_{k\in\mathcal{U}_l^m} \left\{1/H_m^k\right\}=\max_{k\in \emptyset} \left\{1/H_m^k\right\}=0$, and $c_l^m=c_{l-1}^m$. Eq.~(\ref{eq:FoldTerm}) still holds true.  

Finally, the objective function~(\ref{eq:ObjFun2}) can be rewritten as
\begin{equation}
\sum_{m=0}^M N_0 \Gamma_m \sum_{l=1}^L (1 + \Gamma_m)^{c_l^m} \max_{k \in \mathcal{U}_l^m} \left\{1 / H_m^k \right\}.
\end{equation}
Since $(1+\Gamma_m)>0$, it can be seen that to minimize the total BS power, we need to keep the $c_l^m$'s as low as possible.

\subsubsection{Performance Bounds \label{subsec:bounds}}

The reformulation and simplification allow us to derive 
performance bounds for the total BS power consumption.  First, we derive the upper bound for the objective function~(\ref{eq:ObjFun2}). Define a variable
\begin{equation}
  \overline{G}_m = \max_{l\in\{1,\cdots,L\}} \max_{k\in\mathcal{U}_l^m} \left\{\Gamma_m/H_m^k\right\},
\end{equation}
which corresponds to the user with the worst channel condition among all users that connect to BS $m$. It follows that:
\begin{eqnarray}
 \sum_{m=0}^M Q_1^m &\hspace{-0.1in} =& \hspace{-0.1in} N_0 \sum_{m=0}^M \sum_{l=1}^L (1+\Gamma_m)^{c_l^m}\max_{k\in\mathcal{U}_l^m}\left\{\Gamma_m / H_m^k\right\} \nonumber \\
 &\hspace{-0.1in} \le& \hspace{-0.1in} N_0 \sum_{m=0}^M \sum_{l=1}^L (1+\Gamma_m)^{c_l^m} \overline{G}_m \nonumber \\
 &\hspace{-0.1in} \le& \hspace{-0.1in} N_0 \sum_{m=0}^M \overline{G}_m \sum_{l=1}^L (1+\Gamma_m)^{l-1} \nonumber \\
 &\hspace{-0.1in} =& \hspace{-0.1in} N_0 \sum_{m=0}^M \overline{G}_m \left[ (1+\Gamma_m)^L-1 \right] / \Gamma_m. \label{eq:UBound}
\end{eqnarray}

In~(\ref{eq:UBound}), the first inequality is from the definition of $\overline{G}_m$. The second inequality is from the definition of $c_{l+1}^m$. Specifically, $c_1^m=0$; when $\mathcal{U}_l^m \neq \emptyset$, we have $c_{l}^m=c_{l-1}^m+1$; when $\mathcal{U}_l^m = \emptyset$, we have $c_{l}^m=c_{l-1}^m$. It follows that $c_l^m\le l-1$. 
Therefore, (\ref{eq:UBound}) is an upper bound on the objective function~(\ref{eq:ObjFun2}). 

Furthermore, by defining $\overline{G} = \max_{m\in\{0,\cdots,M\}} \left\{ \overline{G}_m \right\}$, and $\overline{\Gamma} = \max_{m\in\{0,\cdots,M\}} \left\{ \Gamma_m \right\}$, we can get a looser upper bound from~\ref{eq:UBound} as
\begin{equation}
 \sum_{m=0}^M Q_1^m \leq N_0 \overline{G} (M+1)\left[ (1+\overline{\Gamma})^L-1\right]/\overline{\Gamma}.
\end{equation}

Next, we derive a lower bound for (\ref{eq:ObjFun2}). Define
\begin{equation}
 \left\{ \begin{array}{l}
  \underline{G}^l = \min_{m\in\{0,\cdots,M\}} \max_{k\in\mathcal{U}_l^m} \left\{ \Gamma_m/ H_m^k \right\} \\
  \underline{\Gamma} = \min_{m\in\{0,\cdots,M\}} \left\{ \Gamma_m \right\}. 
          \end{array} \right.
\end{equation}
We have that
\begin{eqnarray}
 \sum_{m=0}^M Q_1^m &\hspace{-0.1in} =& \hspace{-0.1in} N_0 \sum_{m=0}^M \sum_{l=1}^L (1+\Gamma_m)^{c_l^m} \max_{k\in\mathcal{U}_l^m} \left\{ \Gamma_m / H_m^k \right\} \nonumber \\
 &\hspace{-0.1in} \ge& \hspace{-0.1in} N_0 \sum_{m=0}^M \sum_{l=1}^L (1+\Gamma_m)^{c_l^m} \underline{G}^l \nonumber \\
 &\hspace{-0.1in} \ge& \hspace{-0.1in} N_0 \sum_{l=1}^L \underline{G}^l \sum_{m=0}^M (1+\underline{\Gamma})^{c_l^m} \nonumber \\
 &\hspace{-0.1in} \ge& \hspace{-0.1in} N_0 (M+1) \sum_{l=1}^L \underline{G}^l (1+\underline{\Gamma})^{\frac{\sum_{m=0}^Mc_l^m}{M+1}}\nonumber \\
 &\hspace{-0.1in} \ge& \hspace{-0.1in} N_0  (M+1) \sum_{l=1}^L \underline{G}^l(1+\underline{\Gamma})^{\frac{l-1}{M+1}}. \label{eq:LBound}
\end{eqnarray}

In~(\ref{eq:LBound}), the first inequality is from the definition of $\underline{G}^l$. The second inequality is due to the definition of $\underline{\Gamma}$. The third inequality is due to the fact that $(1+\Gamma)^{c_l^m}$ is a convex function. The fourth inequality is because that each level must be transmitted by at least one BS. Thus for each level $l$, there is at least one $c_l^m=c_{l-1}^m+1$ for some $m$.  It follows that the sum $\sum_{m=0}^M c_l^m$ should be greater than $l-1$. Therefore, (\ref{eq:LBound}) provides a lower bound for~(\ref{eq:ObjFun2}). 

Furthermore, by defining $\underline{G} = \min_{l\in\{1,\cdots,L\}} \left\{ \underline{G}^l \right\}$, we can obtain a looser lower bound from (\ref{eq:LBound}) as
\begin{equation}
 \sum_{m=0}^M Q_1^m \geq N_0 \underline{G} (M+1) \frac{(1+\underline{\Gamma})^{\frac{L}{M+1}}-1}{(1+\underline{\Gamma})^{\frac{1}{M+1}}-1}.
\end{equation}

\subsubsection{A Simple Heuristic Scheme \label{subsec:heuristic}}

We first describe a greedy heuristic algorithm that solves OPT-Power with suboptimal solutions.  With this heuristic, each user compares the channel gains from the MBS and the FBS's. It chooses the BS with the best channel condition to connect to, thus the values of the binary variables $I_m^k$ are determined.  Once the binary variables are fixed, all the subsets $\mathcal{U}_l^m$'s are determined. Starting with $Q_{L+1}^m=0$, we can apply (\ref{eq:QmlEqu}) iteratively to find the $Q_l^m$'s. Finally, the transmit powers $P_l^m$ can be computed using (\ref{eq:QrepP}). 


With this approach, among the users requesting the level $l$ packet, it is more likely that some of them connect to the MBS and the rest connect to some FBS's, due to the random channel gains in each time slot. In this situation, both MBS and FBS will have to transmit all the requested data packets. Such situation is not optimal for minimizing the total power, as will be discussed in Section~\ref{subsubsec:caseII}.

\subsubsection{Power Allocation Algorithms \label{subsec:proposed}}

In the following, we develop three power allocation algorithms for three different connection scenarios with a more structured approach. 

\paragraph{Case I--One Base Station}

We first consider the simplest connection scenario where all the $K$ users connect to the same BS (i.e., either the MBS or an FBS). Assume all the users connect to BS $m$. Then we have $I_m^k=1$ for all $k$, and all the subsets $\mathcal{U}_l^m$ are non-empty; $I_{m'}^k=0$ for all $k$ and all $m' \neq m$, and all the subsets $\mathcal{U}_l^{m'}$ are empty for $m' \neq m$.

From (\ref{eq:FmDef}), we can derive the optimal solution as:
\begin{eqnarray}\label{eq:OptCase1}
Q_l^{m\ast} &=& (1+\Gamma_m) Q_{l+1}^{m\ast} + \max_{k \in \mathcal{U}_l^m} \left\{ N_0\Gamma_m / H_m^k \right\}, \nonumber \\
 &=& N_0 \Gamma_m \sum_{i=l}^L (1+\Gamma_m)^{i-l} \max_{k \in \mathcal{U}_l^m} \left\{1/H_m^k\right\}, 
 \;\; l=1,2, \cdots,L. 
\end{eqnarray}
Recall that $Q_{L+1}^{m\ast} = Q_{L+1}^{m} = 0$, the optimal power allocation for Problem OPT-Power in this case is:
\begin{equation}
  P_l^{m'\ast}=\left\{ \begin{array}{ll}
     Q_l^{m\ast}-Q_{l+1}^{m\ast}, & m'=m, \mbox{ for all } l \\
     0, & m' \neq m, \mbox{ for all } l.
                      \end{array} \right.
\end{equation}

\paragraph{Case II--MBS and One FBS \label{subsubsec:caseII}}

We next consider the case with one MBS and one FBS (i.e., $M=1$), where each user has two choices: connecting to either the FBS or the MBS. 

Recall that $\mathcal{U}_l^0$ and $\mathcal{U}_l^1$ are the subset of users who connected to the MBS and the FBS, respectively, and who request the level $l$ packet. Examining (\ref{eq:FoldTerm}), we find that the total power of BS $m$ can be significantly reduced if one or more levels are not transmitted, since the exponent $c_l^m$ will not be increased in this case. Furthermore, consider the two choices: (i) not transmitting level $l$, and (ii) not transmitting level $l'>l$ from BS $m$. The first choice will yield larger power savings, since more exponents (i.e., $c_l^m, c_{l+1}^m, \cdots, c_{l'-1}^m$) will assume smaller values. 
Therefore, we should let these two subsets be empty whenever possible, i.e., either $\mathcal{U}_l^0=\emptyset$ or $\mathcal{U}_l^1=\emptyset$. According to this policy, all the users requesting the level $l$ packet will connect to the same BS. We only need to make the optimal connection decision for each subset of users requesting the same level of packet, rather than for each individual user. 

Since not transmitting a lower level packet yields more power savings for a BS,
we calculate the power from the lowest to the highest level, and decide whether connecting to the MBS or the FBS for users in each level. Define $G_l^0 = \max_{k\in\mathcal{U}_l} \left\{ 1/H_0^k \right\}$ and $G_l^1 = \max_{k\in\mathcal{U}_l} \left\{ 1/H_1^k \right\}$. The algorithm for solving Problem OPT-Power in this case is given in Table~\ref{tab:Case2Algo}.  In Steps $2$--$10$, the decision 
on whether connecting to the MBS or the FBS is made by comparing the expected increments in the total power. The user subsets $\mathcal{U}_l^0$ and $\mathcal{U}_l^1$ are determined in Steps $4$ and $7$. In Steps $11$--$14$, $Q_l^m$'s and the corresponding $P_l^m$'s are computed in the reverse order, based on the determined subsets $\mathcal{U}_l^0$ and $\mathcal{U}_l^1$. 

The computational complexity of this algorithm is $\mathcal{O}(L)$.

\begin{table}[!t]
\begin{center}
\caption{Power Allocation Algorithm For Case II}
\begin{tabular}{ll}
\hline
1: & Initialize all $c_l^0$, $c_l^1$, $Q_{L+1}^0$ and $Q_{L+1}^1$ to zero; \\
2: & FOR $l=1$ TO $L$ \\
3: & $\;\;$ IF ($\Gamma_0(1+\Gamma_0)^{c_l^0}G_l^0 \le \Gamma_1(1+\Gamma_1)^{c_l^1}G_l^0$) \\
4: & $\;\;\;\;$ Set $\mathcal{U}_l^0=\mathcal{U}_l$ and $\mathcal{U}_l^1=\emptyset$; \\
5: & $\;\;\;\;$ $c_l^0=c_l^0+1$; \\
6: & $\;\;$ ELSE \\
7: & $\;\;\;\;$ Set $\mathcal{U}_l^0=\emptyset$ and $\mathcal{U}_l^1=\mathcal{U}_l$; \\
8: & $\;\;\;\;$ $c_l^1=c_l^1+1$; \\
9: & $\;\;$ END IF\\
10: & END FOR\\
11: & FOR $l=L$ TO $1$ \\
12: & $\;\;$ $Q_l^0=F_0(Q_{l+1}^0,\mathcal{U}_l^0)$ and $P_l^0=Q_l^0-Q_{l+1}^0$; \\
13: & $\;\;$ $Q_l^1=F_1(Q_{l+1}^1,\mathcal{U}_l^1)$ and $P_l^1=Q_l^1-Q_{l+1}^1$; \\
14: & END FOR\\
\hline
\end{tabular}
\label{tab:Case2Algo}
\end{center}
\vspace{-0.2in}
\end{table}

\paragraph{Case III--MBS and Multiple FBS's \label{subsubsec:case3}}

Finally, we consider the general case with one MBS and multiple FBS's in the network. Each user is able to connect to the MBS or a nearby FBS. Recall that we define $\mathcal{U}_l$ as the set of users requesting the level $l$ packet, and $\mathcal{U}_l^m$ as the subset of users in $\mathcal{U}_l$ that {\em connect} to BS $m$.  These sets have the following properties.
\begin{equation}\label{eq:SetProp1}
\left\{ \begin{array}{l}
  \bigcup_{m=0}^M \mathcal{U}_l^m = \mathcal{U}_l  \nonumber \\
  \mathcal{U}_l^m \bigcap \mathcal{U}_l^{m'} = \emptyset, \; \mbox{ for all } m' \neq m. \nonumber
        \end{array} \right.
\end{equation}
The first property is due to the fact that each user must connect to the MBS or an FBS. The second property is because each user can connect to only one BS. The user subsets connecting to different BS's do not overlap. Therefore, $\mathcal{U}_l^m$'s is a {\em partition} of $\mathcal{U}_l$ with respect to $m$.
 
In addition, we define $\mathcal{S}_l^m$ as the set of possible users that are {\em covered} by BS $m$ and request the level $l$ packet. These sets have the following properties. 
\begin{equation}\label{eq:SetProp2}
\left\{ \begin{array}{l}
\bigcup_{m=1}^M \mathcal{S}_l^m = \mathcal{S}_l^0=\mathcal{U}_l  \nonumber \\
\mathcal{S}_l^m \bigcap \mathcal{S}_l^0 = \mathcal{S}_l^m, \; \mbox{ for all } m \neq 0 \nonumber \\
\mathcal{S}_l^m \bigcap \mathcal{S}_l^{m'} = \emptyset, \; \mbox{ for all } m' \neq m \mbox{ and } m, m'\neq 0. \nonumber
        \end{array} \right. 
\end{equation}
The first property is because all users in each femtocell are covered by the MBS. The second property indicates that the users covered by FBS $m$ are a subset of the users covered by the MBS. The third property shows that the user subsets in different femtocells do not overlap. We can see that the $\mathcal{S}_l^m$'s, for $m=1,\cdots,M$, are also a partition of $\mathcal{U}_l$. 

Define $W_m(\mathcal{U})=\max_{k\in\mathcal{U}}\left\{ 1/H_m^k \right\}$, where $\mathcal{U}$ is the set of users and $m=0,\cdots,M$. If the set $\mathcal{U}$ is empty, we define $W_m(\emptyset)=0$. For example, consider Case II where $M=1$.  We have  $\mathcal{S}_l^0 = \mathcal{S}_l^1 = \mathcal{U}_l$, $W_0(\mathcal{U}_l)=G_l^0$, and $W_1(\mathcal{U}_l)=G_l^1$.

The power allocation algorithm for Case III is presented in Table~\ref{tab:Case3Algo}. The algorithm iteratively picks users from the {\em eligible} subset $\mathcal{S}_l^m$ and assigns them to the {\em allocated} subset $\mathcal{U}_l^m$.  
In each step $l$, $\Psi$ is the subset of FBS's that will transmit the level $l$ packet; the complementary set $\overline{\Psi}$ is the subset of FBS's that will not transmit the level $l$ packet. The expected increment in total power for each partition is computed, and the partition with the smallest expected increment will be chosen. $\Delta_l^m$ is the power of BS $m$ for transmitting the level $l$ data packet. In Steps $6$--$15$, the MBS and FBS combination $\Psi$ is determined for transmitting the level $l$ packet, with the lowest power $\Delta_0$.  In Steps $16$--$30$, elements in $\mathcal{S}_l^m$ are assigned to $\mathcal{U}_l^m$ according to $\Psi$. In Steps $31$--$35$, power sums $Q_l^m$ and the corresponding power allocations $P_l^m$ are calculated in the reverse order from the known $\mathcal{U}_l^m$'s. 

The complexity of the algorithm is $\mathcal{O}(ML)$. 

\begin{table}[!t]
\begin{center}
\caption{Power Allocation Algorithm For Case III} \begin{tabular}{ll} \hline
1: & Initialize: $c_l^m=0$ and $Q_{L+1}^m=0$, for all $l$, $m$; \\
2: & FOR $l=1$ TO $L$ \\
3: & $\;\;\;$ FOR $m=0$ TO $M$ \\
4: & $\;\;\;\;\;\;$
$\Delta_l^m=\Gamma_m(1+\Gamma_m)^{c_l^m}W_m(\mathcal{S}_l^m)$; \\
5: & $\;\;\;$ END FOR\\
6: & $\;\;\;$ Set $\Omega=\{1,\cdots,M\}$ and $\Psi=\emptyset$; \\
7: & $\;\;\;$ WHILE ($\Omega \neq \emptyset$) \\
8: & $\;\;\;\;\;\;$ $m'=\arg\min_{m\in\Omega} \Delta_l^m$; \\
9: & $\;\;\;\;\;\;$ Compute
$\Delta'=\Gamma_0(1+\Gamma_0)^{c_l^0}W_0(\bigcup_{m\in\overline{\Psi\cup m'}}\mathcal{S}_l^m)$;
\\
10:& $\;\;\;\;\;\;$ IF ($(\sum_{m\in\Psi\cup m'}\Delta_l^m + \Delta') < \Delta_0$) \\
11:& $\;\;\;\;\;\;\;\;\;$ Add $m'$ to $\Psi$; \\
12:& $\;\;\;\;\;\;\;\;\;$ $\Delta_0=\sum_{m\in\Psi}\Delta_l^m + \Delta'$; \\
13:& $\;\;\;\;\;\;$ END IF \\
14:& $\;\;\;\;\;\;$ Remove $m'$ from $\Omega$; \\
15:& $\;\;\;$ END WHILE \\
16:& $\;\;\;$ IF ($\Psi = \emptyset$) \\
17:& $\;\;\;\;\;\;$ $\mathcal{U}_l^0=\mathcal{S}_l^0$; \\
18:& $\;\;\;\;\;\;$ $c_l^0=c_l^0+1$; \\
19:& $\;\;\;\;\;\;$ Set $\mathcal{U}_l^m = \emptyset$, for all $m \neq 0$; \\
20:& $\;\;\;$ ELSE  \\
21:& $\;\;\;\;\;\;$ $\mathcal{U}_l^0=\bigcup_{m\in\overline{\Psi}}\mathcal{S}_l^m$; \\
22:& $\;\;\;\;\;\;$ IF ($|\Psi|<M$) \\
23:& $\;\;\;\;\;\;\;\;\;$ $c_l^0=c_l^0+1$; \\
24:& $\;\;\;\;\;\;$ END IF\\
25:& $\;\;\;\;\;\;$ FOR $m \in \Psi$ \\
26:& $\;\;\;\;\;\;\;\;\;$ $c_l^m=c_l^m+1$; \\
27:& $\;\;\;\;\;\;\;\;\;$ $\mathcal{U}_l^m=\mathcal{S}_l^m$; \\
28:& $\;\;\;\;\;\;$ END FOR\\
29:& $\;\;\;$ END IF \\
30:& END FOR \\
31:& FOR $l=L$ TO $1$ \\
32:& $\;\;\;$ FOR $m=0$ TO $M$ \\
33:& $\;\;\;\;\;\;$ $Q_l^m=F_m(Q_{l+1}^m,\mathcal{U}_l^m)$ and $P_l^m=Q_l^m-Q_{l+1}^m$; \\
34:& $\;\;\;$ END FOR \\
35:& END FOR\\
\hline
\end{tabular}
\label{tab:Case3Algo}
\end{center}
\vspace{-0.15in}
\end{table}

\subsection{Performance Evaluation \label{sec:sim3}}

We evaluate the performance of the proposed power allocation algorithms using MATLAB\textsuperscript{\scriptsize TM}. Three scenarios corresponding to the three cases in Section~\ref{sec:alg} are simulated: (i) Case I: a single MBS; (ii) Case II: one MBS and one FBS; and (iii) Case III: one MBS and three FBS's.

Since we do not find any similar schemes in the literature, we made the following comparisons. First, we compare 
Cases I and II with respect to BS power consumption and interference footprint.  In both cases, there are $K=8$ users and $L=4$ levels. In Case I, the MBS bandwidth is $B_0=2$ MHz. In Case II, the MBS and the FBS share the $2$ MHz total bandwidth; the MBS bandwidth is $B_0=1$ MHz and the FBS bandwidth is $B_1=1$ MHz. 
The target data rate $R_{tar}$ is set to $2$ Mbps. The channel gain from a base station to each user is exponentially distributed in each time slot. 


The interference footprints in the three dimensional space are plotted in Fig.~\ref{fig:Case1VSCase2}. 
The height $B$ of the cylinders indicates the spectrum used by a BS, 
while the radius $r$ is proportional to the BS transmit power. 
In Case I when only the MBS is used, the total BS power is $45.71$ dBm and the volume of the cylinder is $\pi r^2 B = 18,841$ MHz m$^2$.  In Case II when both the MBS and FBS are used, the total BS power is $34.58$ dBm and the total volume of the two cylinders is $2,378$ MHz m$^2$. Using an additional FBS achieves a $11.13$ dB power saving and the interference footprint is reduced to $12.62$\% of that in Case I. This simple comparison clearly demonstrate the advantages of femtocells achieved by bringing BS's closer to users.

\begin{figure} [!t]
\centering
\includegraphics[width=4.5in, height=3.0in]{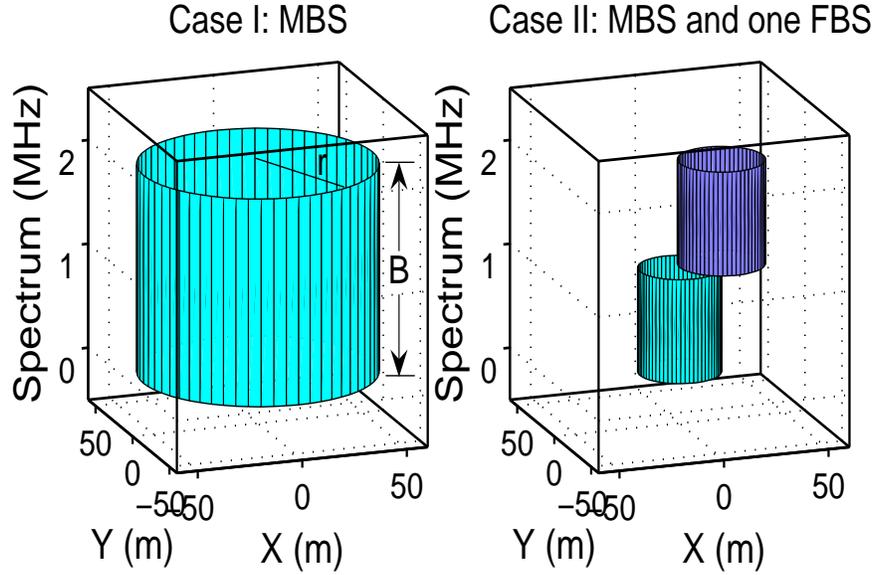}
\caption{Case I vs. Case II: interference footprints.}
\label{fig:Case1VSCase2}
\end{figure}

We next consider the more general Case III, using a femtocell network of one MBS and three FBS's. The MBS bandwidth is $B_0=1$ MHz and each FBS is assigned with bandwidth $B_m=1$ MHz, $m=1, 2, 3$. The target data rate is still $2$ Mbps. In Figs.~\ref{fig:Case3Level} and~\ref{fig:Case3BW0}, we plot four curves, each obtained with: (i) the heuristic scheme described in Section~\ref{subsec:heuristic}; (ii) The proposed algorithm presented in Section~\ref{subsubsec:case3}; (iii) The upper bound; and (iv) the lower bound derived in Section~\ref{subsec:bounds}. Each point in the figures is the average of $10$ simulation runs. The $95\%$ confidence intervals are plotted as error bars, which are all negligible. 

In Fig. \ref{fig:Case3Level}, we examine the impact of the number of packet levels $L$ on the total BS transmit power. We increase $L$ from $2$ to $6$, and plot the total power of base stations. As expected, the more packet levels, the larger the BS power consumption. 
Both the proposed and heuristic curves lie in between the upper and lower bound curves.  When $L$ is increased from $2$ to $6$, the power consumption of the heuristic scheme is increased by $12.22$ dB, while the power consumption of the proposed algorithm is increased by $9.94$ dB. The power savings achieved by the proposed algorithm over the heuristic scheme range from $3.92$ dB to $6.45$ dB. 

\begin{figure}[!t]
\centering
\includegraphics[width=4.5in, height=3.0in]{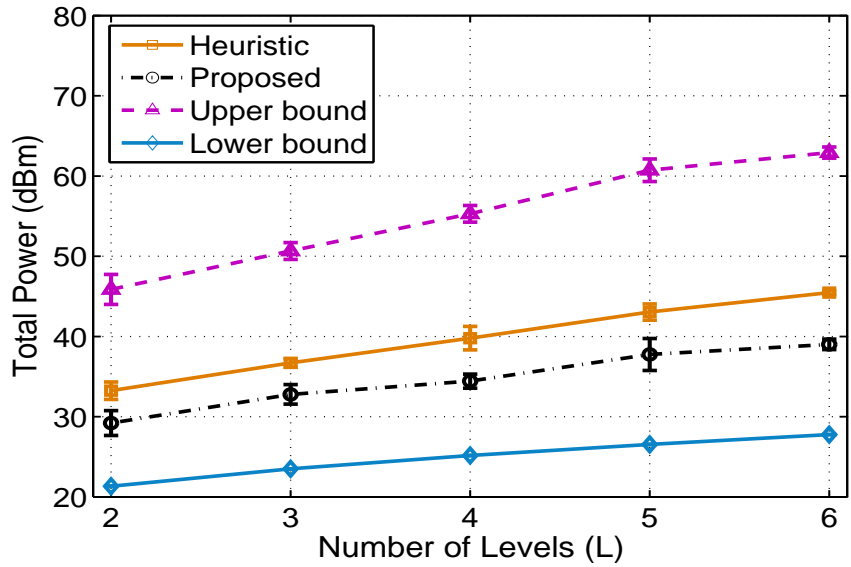}
\caption{Case III: impact of number of levels $L$.}
\label{fig:Case3Level}
\end{figure}

In Fig. \ref{fig:Case3BW0}, we show the impact of the BS bandwidths. The number of levels is $L=4$. We fix the total bandwidth at $2$ MHz, which is shared by the MBS and FBS's. We increase the MBS bandwidth from $0.4$ MHz to $1.6$ MHz in steps of $0.2$ MHz, while decrease the bandwidth of FBS's from $1.6$ MHz to $0.4$ MHz. We find that the total power consumption is increased as $B_0$ gets large. This is due to the fact that as the FBS bandwidth gets smaller, the FBS's have to spend more power to meet the minimum data rate requirement. The curve produced by the proposed algorithm has a smaller slop than that of the heuristic scheme: the overall increase in the total power of the proposed algorithm is $4.86$ dB, while that of the heuristic scheme is $20.84$ dB. This implies that the proposed scheme is not very sensitive to the bandwidth allocation between the MBS and FBS's.  The proposed algorithm also achieves consider power savings over the heuristic scheme. When $B_0=1.6$ MHz, the total power of the proposed algorithm is $20.75$ dB lower than that of the heuristic scheme. 

\begin{figure}[!t]
\centering
\includegraphics[width=4.5in, height=3.0in]{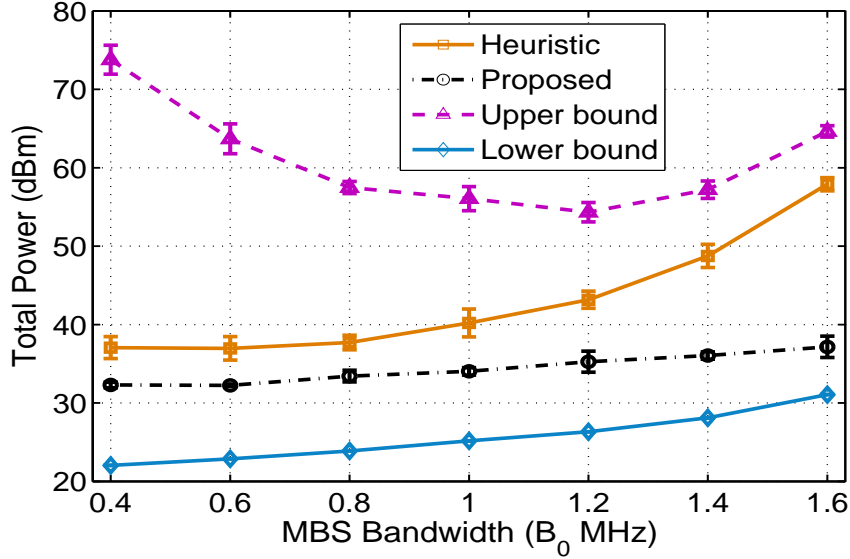}
\caption{Case III: impact of MBS bandwidth $B_0$.}
\label{fig:Case3BW0}
\end{figure}


\section{Video over CR Femtocell Networks}\label{sec:femto_cr_video}
In this section, we investigate the problem of video streaming in femtocell cognitive radio (CR) networks and formulate a {\em stochastic programming} problem to examine three deployment scenarios.  In the case of a single FBS, we apply {\em dual decomposition} to develop an optimum-achieving distributed algorithm, which is shown also optimal for the case of multiple non-interfering FBS's. In the case of multiple interfering FBS's, we develop a greedy algorithm that can compute near-optimal solutions, and prove a closed-form lower bound for its performance based on an {\em interference graph} model. The proposed algorithms are evaluated with simulations, and are shown to outperform three alternative schemes with considerable gains. 

\subsection{System Model and Preliminaries \label{sec:mod4}}

\subsubsection{Spectrum and Network Model}

We consider a spectrum consisting of $(M+1)$ channels,  including one common, unlicensed channel (indexed as channel $0$) and $M$ licensed channels (indexed as channels $1$ to $M$).  The $M$ licensed channels are allocated to a primary network, and the common channel is exclusively used by all CR users. We assume all the channels follow a synchronized time slot structure~\cite{Zhao07a}. The capacity of each licensed channel is $B_1$ Mbps, while the capacity of the common channel is $B_0$ Mbps. The channel states evolve independently, while the occupancy of each licensed channel follows a two-state discrete-time Markov process.

The femtocell CR network is illustrated in Fig.~\ref{fig:netmod2}. There is an MBS and $N$ FBS's deployed in the area to serve CR users.  The $N$ FBS's are connected to the MBS (and the Internet) via broadband wireline connections. Due to advances in antenna technology,
it is possible to equip multiple antennas at the base stations.  The MBS has one antenna that is always tuned to the common channel.  Each FBS is equipped with multiple antennas (e.g., $M$) and is able to sense multiple licensed channels at the beginning of each time slot.  There are $K_i$ CR users in femtocell $i$, $i=1,2,\cdots,N$, and $\sum_{i=1}^N K_i = K$.  Each CR user has a software radio transceiver, which can be tuned to any of the $M$+1 channels.  A CR user will either connect to a nearby FBS using one or more of the licensed channels or to the MBS via the common channel. 

\begin{figure}[!t]
\centering
\includegraphics[width=4.0in, height=2.7in]{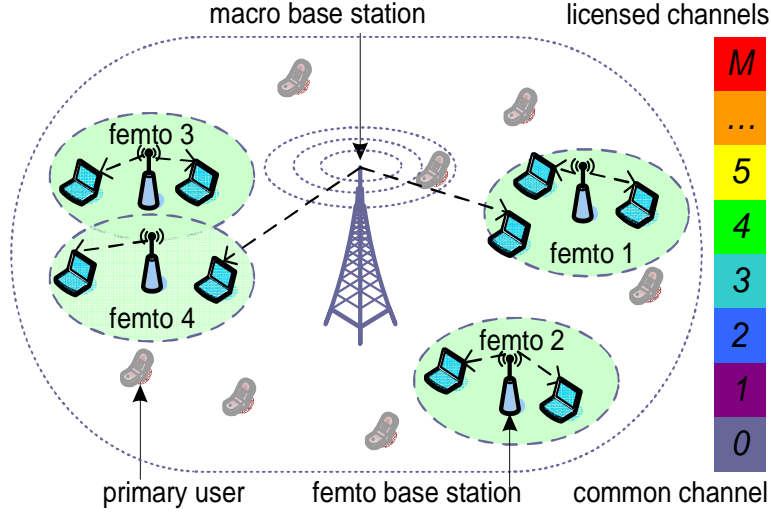}
\caption{A femtocell CR network with one MBS and four FBS's.}
\label{fig:netmod2}
\vspace{-0.1in}
\end{figure}

Although the CR users are mobile, we assume constant topology during a time slot. If the topology is changed during a time slot, the video transmission will only be interrupted for the time slot, since the proposed algorithms are executed in every time slot for new channel assignment and schedule.

\subsubsection{Spectrum Sensing and Access}

The femtocell CR network is within the coverage of the infrastructure-based primary network.  Both FBS's and CR users sense the channels to identify spectrum opportunities in each time slot.  Each time slot consists of (i) a {\em sensing phase}, when CR users and FBS's sense licensed channels, (ii) a {\em transmission phase}, when CR users and FBS's attempt to access licensed channels, 
and (iii) an {\em acknowledgment phase}, when acknowledgments (ACK) are returned to the source. 

Cooperative sensing policy is also adopted here. We also adopt a {\em hypothesis test} to detect channel availability. 
We assume that each CR user chooses one channel to sense in each time slot, since it only has one transceiver.  The sensing results will be shared among CR users and FBS's via the common channel in the sensing phase. Given $L$ sensing results on channel $m$, the availability of channel $m$, i.e., $P^A_m(\Theta_1^m,\cdots,\Theta_L^m)$, can be computed iteratively as follows.
\begin{eqnarray}\label{eq:Iteration2}
P^A_m(\Theta_1^m)&=&\left[1+\frac{\eta_m}{1-\eta_m}\times\frac{(\delta_1^m)^{1-\Theta_1^m}(1-\delta_1^m)^{\Theta_1^m}}{(\epsilon_1^m)^{\Theta_1^m}(1-\epsilon_1^m)^{1-\Theta_1^m}}\right]^{-1} \\ 
P^A_m(\vec{\Theta}_l^m)&=&P^A_m(\Theta_1^m,\Theta_2^m,\cdots,\Theta_l^m) \nonumber \\
&=&\left\{1+\left[\frac{1}{P^A_m(\Theta_1^m,\Theta_2^m,\cdots,\Theta_{l-1}^m)}-1\right]\times 
\right.\nonumber \\
&&\left.
\frac{(\delta_l^m)^{1-\Theta_l^m}(1-\delta_l^m)^{\Theta_l^m}}{(\epsilon_l^m)^{\Theta_l^m}(1-\epsilon_l^m)^{1-\Theta_l^m}}\right\}^{-1},l=2,\cdots,L. 
\end{eqnarray}

We adopt a probabilistic approach: based on sensing results $\vec{\Theta}_m$, we have $D_m(t)=0$ with probability $P^D_m(\vec{\Theta}_m)$ and $D_m(t)=1$ with probability $1-P^D_m(\vec{\Theta}_m)$. For primary user protection, the collision probability with primary users caused by CR users should be bounded. The probability $P^D_m(\vec{\Theta}_m)$ is determined as follows
\begin{equation}\label{eq:PrDm2}
P^D_m(\vec{\Theta}_m)=\min \left\{ \gamma_m / \left[ 1 - P^A_m(\vec{\Theta}_m) \right], 1 \right\}.
\end{equation}

Let $\mathcal{A}(t) := \{m|D_m(t)=0\}$ be the set of available channels in time slot $t$. Then $G^t=\sum_{m\in \mathcal{A}(t)} P^A_m(\Theta_1^m)$ is the expected number of available channels.  These channels will be accessed in the transmission phase of time slot $t$.

\subsubsection{Channel Model}

Without loss of generality, we consider independent block fading channels that is widely used in prior work~\cite{Rappaport01}. The channel fading-gain process is piecewise constant on blocks of one time slot, and fading in different time slots are independent.  Let $f^{i,j}_X(x)$ denote the {\em probability density function} of the received SINR $X$ from a base station $i$ 
at CR user $j$.  We assume the packet can be successfully decoded if the received SINR exceeds a threshold $H$.  The packet loss probability from base station $i$ to CR user $j$ is 
\begin{equation}\label{eq:PrFad}
P_{i,j}=\Pr\{X \le H\} = \int_0^{H}f_X^{i,j}(x) dx=F_X^{i,j}(H),
\end{equation}
where $F_X^{i,j}(H)$ is the cumulative density function of $X$.

In the case of correlated fading channels, which can be modeled as finite state Markov Process~\cite{Zhang99}, the packet loss probability in the next time slot can be estimated from the known state of the previous time slot and the transition probabilities. 
If the packet is successfully decoded, the CR user returns an ACK to the base station in the ACK phase.  We assume ACKs are always successfully delivered.

\subsubsection{Video Performance Measure}

We assume each active CR user receives a real-time video stream from either the MSB or an FSB.  Without loss of generality, we adopt the 
MGS option of H.264/SVC, 
for scalability to accommodate the high variability of network bandwidth in CR networks.  

Due to real-time constraint, each Group of Pictures (GOP) of a video stream must be delivered in the next $T$ time slots. With MGS, enhancement layer NAL units can be discarded from a quality scalable bit stream, and thus packet-based quality scalable coding is provided. Our approach is to encode the video according to the maximum rate the channels can support. During transmission, 
only part of the MGS video gets transmitted as allowed by the current available channel bandwidth. The video packets are transmitted in decreasing order of their significance in decoding. 
When a truncated MGS video is received and decoded, the PSNR is computed by substituting the effective rate of the received MGS video into (\ref{eq:QuaMod}) given below, thus the original video is not required.

Without loss of generality, we assume that the last wireless hop is the bottleneck; video data is available at the MBS and FBS's when they are scheduled to be transmitted.  The quality of reconstructed MGS video can be modeled as~\cite{Wien07}:
\begin{equation}\label{eq:QuaMod}
  W(R)=\alpha+\beta \times R, 
\end{equation}
where $W(R)$ is the average peak signal-to-noise ratio (PSNR) of the reconstructed video, $R$ is the received data rate, $\alpha$ and $\beta$ are constants depending on the video sequence and codec. 

We verified (\ref{eq:QuaMod}) using an H.264/SVC codec and the {\em Bus}, {\em Mobile}, and {\em Harbour} test sequences.  In Fig.~\ref{fig:mgs-rd}, the markers are obtained by truncating the encoded video's enhancement layer at different positions to obtain different effective rates, while the curves are computed using (\ref{eq:QuaMod}). The curves fit well with measurements for the three sequences. It is worth noting that PSNR may not be a good measure of video quality as compared with alternative metrics such as MS-SSIM~\cite{Wang04}. The main reason for choosing PSNR is that there is a closed-form model relating it to network level metrics--video rate. With the closed-form model, we can have a mathematical formulation of the scheduling/resource allocation problem, and derive effective algorithms. Should such closed-form models be available for MS-SSIM, it is possible to incorporate it into the optimization framework as well. 

\begin{figure}[!t]
\centering
\includegraphics[width=4.5in, height=3.0in]{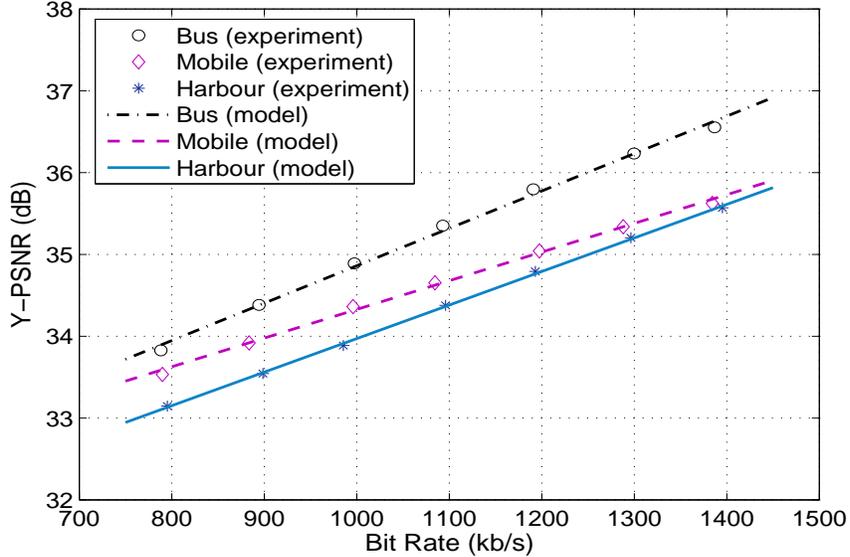}
\caption{Rate-distortion curves of three H.264/SVC MGS videos.}
\label{fig:mgs-rd}
\vspace{-0.1in}
\end{figure}

\subsection{MGS Video over Femtocell CR Networks \label{sec:alg2}}

In this section, we address the problem of resource allocation for MGS videos over femtocell CR networks.  We first examine the case of a single FBS, and then the more general case of multiple non-interfering or interfering FBS's.  The algorithms for the single and non-interfering FBS cases are distributed ones and optimal. The algorithm for the interfering FBS case is a centralized one that can be executed at the MBS. To simplify notation, we omit the time slot index $t$ for most of the variables in this Section. For example, $x$ represents a variable for time slot $t$, $x^{-}$ represents the variable in time slot $(t-1)$, and $x^{+}$ represents the variable in time slot $(t+1)$.

\subsubsection{Case of Single FBS}

\paragraph{Formulation}

We first consider the case of a single FBS in the CR network, where the FBS can use all the $G$ available channels to stream videos to $K$ active CR users. Let $w_j$ be the PSNR of CR user $j$ at the beginning of time slot $t$ and $W_j$ the PSNR of CR user $j$ at the end of time slot $t$.  In time slot $t$, $w_j$ is already known; $W_j$ is a random variable that depends on channel condition and primary user activity; and $w_j^{+}$ is a {\em realization} of $W_j$.  Let $\xi_{0,j}$ and $\xi_{1,j}$ indicate the random packet losses from the MBS and FBS, respectively, to CR user $j$ in time slot $t$. 
That is, $\xi_{i,j}$ is $1$ with probability $\bar{P}_{i,j}=1-P_{i,j}$ and $0$ with probability $P_{i,j}$. Due to block fading channels, $P_{i,j}$'s do not change within the time slot. 

Let $\rho_{0,j}$ and $\rho_{1,j}$ be the portions of time slot $t$ when CR user $j$ receives video data from the MBS and FBS, respectively.  
The average PSNR is computed every $T$ time slots. We first have $W_j(0)=\alpha_j$, when $t=0$. In each time slot $t$, the CR user receives $\xi_{0,j} \rho_{0,j} B_0$ bits through the MBS, and $\xi_{1,j} \rho_{1,j} G B_1$ bits through the FBS (assuming that OFDM is used), which contribute an increase of $\beta (\xi_{0,j} \rho_{0,j} B_0 + \xi_{1,j} \rho_{1,j} G B_1) / T$ to the total PSNR in this $T$ time slot interval, according to~(\ref{eq:QuaMod}). Therefore we have the following recursive relationship:
$W_j = W_j^{-} + \beta (\xi_{0,j} \rho_{0,j} B_0 + \xi_{1,j} \rho_{1,j} G B_1) / T$ = $W_j^{-} + \xi_{0,j} \rho_{0,j} R_{0,j} + \xi_{1,j} \rho_{1,j} G R_{1,j}$, where $R_{0,j}=\beta B_0/T$ and $R_{1,j}=\beta B_1/T$. 

For proportional fairness, we aim to maximize the sum of the logarithms of the PSNRs of all CR users~\cite{Kelly98}.  We formulate a {\em multistage stochastic programming problem} by maximizing the {\em expectation} of the logarithm-sum at time $T$.
\begin{eqnarray}\label{eq:MultStage}
 \mbox{maximize:} && \hspace{-0.2in} \sum_{j=1}^K \mathbb{E}[\log(W_j(T))] \\
 \mbox{subject to:} && \hspace{-0.2in} W_j=W_j^{-}+\xi_{0,j} \rho_{0,j} R_{0,j} + \xi_{1,j} \rho_{1,j} G R_{1,j}, 
 \;\; j=1,\cdots,K, \; t=1,\cdots,T \nonumber \\
 && \hspace{-0.2in} \sum_{j=1}^K \rho_{i,j} \leq 1, \;\; i=0,1, \; t=1,\cdots,T \nonumber \\
 && \hspace{-0.2in} \rho_{i,j} \ge 0, \; i=0,1, \; j=1,\cdots,K, \; t=1,\cdots,T.   \nonumber
\end{eqnarray}
$R_{0,j}=\beta_j B_0/T$ and $R_{1,j}=\beta_jB_1/T$ are constants for the $j$-th MGS video. 

At the beginning of the last time slot $T$, a realization $\bm{\xi}_{[T-1]} = [ \vec{\xi}_1, \vec{\xi}_2, \cdots, \vec{\xi}_{T-1} ]$ is known, where $\vec{\xi}_t = [ \xi_{0,1}^t, \xi_{0,2}^t, \cdots, \xi_{0,K}^t, \xi_{1,1}^t, \cdots, \xi_{1,K}^t ]$, $t = 1, 2, \cdots, T-1$. It can be shown that the multistage stochastic programming problem (\ref{eq:MultStage}) can be decomposed into $T$ serial sub-problems, each to be solved in a time slot $t$~\cite{Hu10TW}.
\begin{eqnarray}\label{eq:SingStage}
 \mbox{maximize:} && \hspace{-0.2in} \sum_{j=1}^K \mathbb{E}\{\log(W_j)|\bm{\xi}_{[t-1]}\} \\ 
 \mbox{subject to:} && 
 \hspace{-0.2in} W_j=W_j^{-}+\xi_{0,j} \rho_{0,j} R_{0,j} + \xi_{1,j} \rho_{1,j} G R_{1,j}, 
 j=1,\cdots,K \nonumber \\
 && \hspace{-0.2in} \sum_{j=1}^K \rho_{i,j} \le 1, \;\; i=0,1 \nonumber \\
 && \hspace{-0.2in} \rho_{i,j} \ge 0, \; i=0,1, \; j=1,\cdots,K,   \nonumber
\end{eqnarray}
where $\mathbb{E}\{\log(W_j)|\bm{\xi}_{[t-1]}\}$ denotes the {\em conditional expectation} of $\log(W_j)$ given realization $\bm{\xi}_{[t-1]}$.  $W_j^{-}$ is known given the realization. When $t=1$, the conditional expectation becomes an unconditional expectation. 

Since a CR user has only one transceiver, it can operate on either one or more licensed channels (i.e., connecting to the FBS) or the common channel (i.e., connecting to the MBS), but not both simultaneously.  Assume CR user $j$ operates on the common channel with probability $p_j$ and one or more licensed channels with probability $q_j$.  We then rewrite problem (\ref{eq:SingStage}) as
\begin{eqnarray}\label{eq:ProbOpt1}
 \mbox{maximize:} && \hspace{-0.2in} \sum_{j=1}^K \left[ p_j \bar{P}_{0,j} \log(W_j^{-}+\rho_{0,j} R_{0,j}) \; + 
 q_j \bar{P}_{1,j} \log(W_j^{-} + \rho_{1,j} G R_{1,j}) \right] \nonumber \\
 \mbox{subject to:} && \hspace{-0.2in} 
 \sum_{j=1}^K \rho_{i,j} \le 1, \;\; i=0,1 \nonumber  \\
 && \hspace{-0.2in} p_j + q_j = 1, \;\; j=1,\cdots,K \nonumber \\
 && \hspace{-0.2in} \rho_{i,j}, \; p_j, \; q_j \ge 0, \;\; i=0,1, \; j=1,\cdots,K. \nonumber 
\end{eqnarray}

\paragraph{Properties}

In this section, we analyze the formulated problem (\ref{eq:ProbOpt1}) and derive its properties.  We have Lemmas 1, 2, and 3 and Theorem 1 and provide the proofs in the following. 

\begin{lemma}
Problem (\ref{eq:ProbOpt1}) is a convex optimization problem.
\end{lemma}

\begin{proof} 
First, it can be shown that the single term 
$
p_j \bar{P}_{0,j} \log(W_j^{-} + \rho_{0,j} R_{0,j}) + q_j \bar{P}_{1,j} \log(W_j^{-} + \rho_{1,j} G R_{1,j})
$ 
is a concave function, because its {\em Hessian matrix} is negative semi-definite.  Then, the objective function is concave since the sum of concave functions is also concave. Finally, all the constraints are linear.  We conclude that problem (\ref{eq:ProbOpt1}) is convex 
with a unique optimal solution.
\end{proof}

\begin{lemma}
If $[\rho,p,q]$ is a feasible solution to problem (\ref{eq:ProbOpt1}), then $[\rho,q,p]$ is also feasible. 
\end{lemma}

\begin{proof}
Since $[\rho,p,q]$ is feasible, we have $p + q =1$.  Switching the two probabilities, we still have $q + p =1$. Therefore, the derived new solution is also feasible. 
\end{proof}

\begin{lemma}
Let the optimal solution be $[\rho^{\ast},p^{\ast},q^{\ast}]$.  If $p_j^{\ast} \ge q_j^{\ast}$, then $\bar{P}_{0,j} \log(W_j^{-} + \rho_{0,j}^{\ast} R_{0,j})$ is greater than or equal to $\bar{P}_{1,j} \log(W_j^{-} + \rho_{1,j}^{\ast} G R_{1,j})$. And vice versa.
\end{lemma}

\begin{proof}
Assume $\bar{P}_{0,j} \log(W_j^{-} + \rho_{0,j}^{\ast} R_{0,j})$ is less than $\bar{P}_{1,j} \log(W_j^{-} + \rho_{1,j}^{\ast} G R_{1,j})$.  Since $p_j^{\ast} \ge q_j^{\ast}$, the sum of the product $p_j^{\ast} \bar{P}_{0,j} \log(W_j^{-} + \rho_{0,j}^{\ast} R_{0,j}) + q_j^{\ast} \bar{P}_{1,j} \log(W_j^{-} + \rho_{1,j}^{\ast} G R_{1,j})$ is smaller than the sum of the product $q_j^{\ast} \bar{P}_{0,j} \log(W_j^{-} + \rho_{0,j}^{\ast} R_{0,j}) + p_j^{\ast} \bar{P}_{1,j} \log(W_j^{-} + \rho_{1,j}^{\ast} G R_{1,j})$. Thus we can obtain an objective value larger than the optimum by switching the values of $p_j^{\ast}$ and $q_j^{\ast}$, which is still feasible according to Lemma 2.  This conflicts with the assumption that $[\rho^{\ast},p^{\ast},q^{\ast}]$ is optimal.  The reverse statement can be proved similarly. 
\end{proof}

\begin{theorem}
Let the optimal solution be $[\rho^{\ast},p^{\ast},q^{\ast}]$.  If $p_j^{\ast} > q_j^{\ast}$, then we have $p_j^{\ast} = 1$ and $q_j^{\ast} = 0$.  Otherwise, we have $p_j^{\ast} = 0$ and $q_j^{\ast} = 1$. 
\end{theorem}

\begin{proof}
If $p_j^{\ast} > q_j^{\ast}$, we have 
$\bar{P}_{0,j} \log(W_j^{-} + \rho_{0,j}^{\ast} R_{0,j}) \geq \bar{P}_{1,j} \log(W_j^{-} + \rho_{1,j}^{\ast} G R_{1,j})$ according to Lemma 3.  Since the objective function is linear with respect to $p_j$ and $q_j$, the optimal value can be achieved by setting $p_j$ to its maximum value 1 and $q_j$ to its minimum value 0.  The reverse statement can be proved similarly. 
\end{proof}

According to Theorem 1, a CR user is connected to either the MBS or the FBS for the {\em entire} duration of a time slot in the optimal solution.  That is, it does not switch between base stations during a time slot under optimal scheduling.

\paragraph{Distributed Solution Algorithm}

To solve problem (\ref{eq:ProbOpt1}), we define non-negative {\em dual variables} ${\mathcal \lambda}=[\lambda_0, \lambda_1]$ for the two inequality constraints. The
{\em Lagrangian function} is
\begin{eqnarray}\label{eq:LagProbOpt1}
\hspace{-0.2in} \mathcal{L}(p,\rho,\lambda) \hspace{-0.025in}&=&\hspace{-0.025in} \sum_{j=1}^K \left[ p_j \bar{P}_{0,j} \log(W_j^{-} + \rho_{0,j} R_{0,j}) + 
(1-p_j) \bar{P}_{1,j} \log(W_j^{-} + \rho_{1,j} G R_{1,j}) \right] +  
\nonumber \\
&& \lambda_0(1 - \sum_{j=1}^K \rho_{0,j})+ \lambda_1(1 - \sum_{j=1}^K \rho_{1,j}) \nonumber \\
&=& \hspace{-0.025in} \sum_{j=1}^K \mathcal{L}_j(p_j,\rho_{0,j},\rho_{1,j},\lambda_0,\lambda_1) \hspace{-0.025in} + \hspace{-0.025in} \lambda_0 \hspace{-0.025in} + \hspace{-0.025in} \lambda_1, 
\end{eqnarray}
where 
\begin{eqnarray}
\hspace{0.1in} \mathcal{L}_j(p_j,\rho_{0,j},\rho_{1,j},\lambda_0,\lambda_1) = p_j \bar{P}_{0,j} \log(W_j^{-} + \rho_{0,j}  R_{0,j}) + \nonumber \\
\hspace{0.2in} (1 - p_j) \bar{P}_{1,j} \log(W_j^{-} + \rho_{1,j} G R_{1,j}) - \lambda_0 \rho_{0,j} - \lambda_1 \rho_{1,j}. \nonumber
\end{eqnarray}

The corresponding problem can be decomposed into $K$ sub-problems and solved iteratively.  In Step $\tau \geq 1$, for given $\lambda_0(\tau)$ and $\lambda_1(\tau)$ values, each CR user $j$ solves the following sub-problem using local information.
\begin{eqnarray}\label{eq:ArgSubOpt1}
[p_j^{\ast}(\tau), \rho_{0,j}^{\ast}(\tau), \rho_{1,j}^{\ast}(\tau)] 
= \stackbin[p_j,\rho_{0,j},\rho_{1,j} \ge 0]{}{\arg\max}   
 \mathcal{L}_j(p_j,\rho_{0,j},\rho_{1,j},\lambda_0(\tau),\lambda_1(\tau)).
\end{eqnarray}
There is a unique optimal solution since the objective function in (\ref{eq:ArgSubOpt1}) is concave.  The CR users then exchange their solutions.  The {\em master dual problem}, for given $p(\tau)$ and $\rho(\tau)$, is:
\begin{eqnarray}\label{eq:MaterOpt1}
\min_{\lambda\ge 0} \mathcal{L}(p(\tau),\rho(\tau),\lambda) 
= \sum_{j=1}^K \mathcal{L}_j(p_j(\tau),\rho_{0,j}(\tau),\rho_{1,j}(\tau),\lambda_0,\lambda_1)+\lambda_0+\lambda_1. 
\end{eqnarray}
Since the Lagrangian function is differentiable, the 
{\em gradient iteration} approach can be used.
\begin{equation}\label{eq:IterOpt1}
\lambda_i(\tau+1) = \left[ \lambda_i(\tau) - s \times \left( 1 - \sum_{j=1}^K \rho_{i,j}^{\ast}(\tau) \right) \right]^+, \; i=0,1,
\end{equation}
where $s$ is a sufficiently small positive {\em step size} and $[\cdot]^+$ denotes the projection onto the nonnegative axis. The updated $\lambda_i(\tau+1)$ will again be used to solve the sub-problems, and so forth. Since the problem is convex, we have {\em strong duality}; the {\em duality gap} between the primal and dual problems is zero. The dual variables $\lambda(\tau)$ will converge to the optimal values as $\tau$ goes to infinity. Since the optimal solution to (\ref{eq:ArgSubOpt1}) is unique, the primal variables $p(\tau)$ and $\rho_{i,j}(\tau)$ will also converge to their optimal values when $\tau$ is sufficiently large.

The distributed solution procedure is presented in Table~\ref{tab:Opt1}. In the table, Steps 3--8 solve the sub-problem in (\ref{eq:ArgSubOpt1}); Step 9 updates the dual variables.  The threshold $\phi$ is a prescribed small value with $0 \leq \phi \ll 1$. The algorithm terminates when the dual variables are sufficiently close to the optimal values. 

\begin{table}[!t]
\begin{center}
\caption{Algorithm for the Case of Single FBS}
\begin{tabular}{ll}
\hline 
1: & Set $\tau=0$, $\lambda_0(0)$ and $\lambda_1(0)$ to some nonnegative value; \\
2: & DO \ \ \% (each CR user $j$ executes Steps 3--8)\\
3: & $\;\;$ $\rho_{0,j}(\tau) \hspace{-0.025in} = \hspace{-0.025in} \left[ \frac{\bar{P}_{0,j}}{\lambda_0(\tau)} \hspace{-0.025in} - \hspace{-0.025in} \frac{W_j^{-}}{R_{0,j}} \right]^+$, $\rho_{1,j}(\tau) \hspace{-0.025in} = \hspace{-0.025in} \left[ \frac{\bar{P}_{1,j}}{\lambda_1(\tau)} \hspace{-0.025in} - \hspace{-0.025in} \frac{W_j^{-}}{R_{1,j} G} \right]^+$; \\
4: & $\;\;$ IF $\left[ \bar{P}_{0,j} \log(W_j^{-}+\rho_{0,j}(\tau) R_{0,j})-\lambda_0(\tau)\rho_{0,j}(\tau) \right] >$ \\
   & $\;\;$ $\left[ \bar{P}_{1,j} \log(W_j^{-} + \rho_{1,j}(\tau) G R_{1,j}) - \lambda_1(\tau) \rho_{1,j}(\tau) \right]$ \\
5: & $\;\;\;\;\;$ Set $p_j(\tau)=1$ and $\rho_{1,j}(\tau)=0$; \\
6: & $\;\;$ ELSE  \\
7: & $\;\;\;\;\;$ Set $p_j(\tau)=0$ and $\rho_{0,j}(\tau)=0$; \\
8: & $\;\;$ END IF \\
9: & $\;\;$ MBS updates $\lambda_i(\tau+1)$ as in (\ref{eq:IterOpt1}); \\
10: & $\;\;$ $\tau=\tau+1$; \\
11: & WHILE $\left( \sum_{i=0}^{1}(\lambda_i(\tau+1)-\lambda_i(\tau))^2 > \phi \right)$ \\
\hline 
\end{tabular}
\label{tab:Opt1}
\end{center}
\end{table}

\subsubsection{Case of Multiple Non-interfering FBS's \label{subsec:mulnifbs}}

We next consider the case of $N>1$ non-interfering FBS's. The coverages of the FBS's do not overlap with each other, as FBS 1 and 2 in Fig.~\ref{fig:netmod2}.  Consequently, each FBS can use all the available licensed channels without interfering other FBS's.  Assume each CR user knows the nearest FBS and is associate with it.  Let $\mathcal{U}_i$ denote the set of CR users associated with FBS $i$.  The resource allocation problem becomes:
\begin{eqnarray}\label{eq:ProbOptDM}
 \mbox{maximize:} && \hspace{-0.2in} \sum_{j=1}^K p_j \bar{P}_{0,j} \log(W_j^{-} + \rho_{0,j} R_{0,j}) + 
 \sum_{i=1}^N \sum_{j\in\mathcal{U}_i} q_j \bar{P}_{i,j} \log(W_j^{-} + \rho_{i,j} G R_{i,j}) \\ 
 \mbox{subject to:} && \hspace{-0.2in} \sum_{j=1}^K \rho_{0,j} \leq 1 \nonumber \\
 && \hspace{-0.2in} \sum_{j\in\mathcal{U}_i} \rho_{i,j} \le 1, \;\; i=1,\cdots,N \nonumber \\
 && \hspace{-0.2in} p_j + q_j = 1, \;\; j=1,\cdots,K \nonumber \\
 && \hspace{-0.2in} \rho_{i,j}, \; p_j, \; q_j \ge 0, \;\; 
 i=1,\cdots,N,\; j=1,\cdots,K. \nonumber
\end{eqnarray}
Since all the available channels can be allocated to each FBS with spatial reuse,
problem (\ref{eq:ProbOptDM}) can be solved using the algorithm in Table~\ref{tab:Opt1} with some modified notation: $\rho_{1,j}(\tau)$ now becomes $\rho_{i,j}(\tau)$ and $\lambda_1(\tau)$ becomes $\lambda_i(\tau)$, $i=1, \cdots, N$.  The dual variables are iteratively updated as
\begin{eqnarray}
&& \hspace{-0.4in} \lambda_0(\tau+1)=\left[\lambda_0(\tau)-s \times \left( 1 - \sum_{j=1}^K \rho_{0,j}^{\ast}(\tau) \right) \right]^+ \label{eq:IterOptM0} \\
&& \hspace{-0.4in} \lambda_i(\tau+1)=\left[\lambda_i(\tau)-s \times \left( 1 - \sum_{j\in \mathcal{U}_i} \rho_{i,j}^{\ast}(\tau) \right) \right]^+, 
\;\; i=1,\cdots,N.  \label{eq:IterOptM1}
\end{eqnarray}
The modified solution algorithm is presented in Table~\ref{tab:OptDisM}. As in the case of single FBS, the algorithm is jointly executed by the CR users and MBS, by iteratively updating the dual variables $\lambda_0(\tau)$ and $\lambda_i(\tau)$'s, and the resource allocations $\rho_{0,j}^{\ast}(\tau)$ and $\rho_{i,j}^{\ast}(\tau)$'s. It can be shown that the distributed algorithm can produce the optimal solution for problem (\ref{eq:ProbOptDM}). 

\begin{table}[!t]
\begin{center}
\caption{Algorithm for the Case of Multiple Non-Interfering FBS's}
\begin{tabular}{ll}
\hline 
1: & Set $\tau=0$, and $\lambda_0(0)$ and $\lambda_i(0)$ to some nonnegative values, for \\
   & all $i$; \\
2: & DO \ \ \% (each CR user $j$ executes Steps 3--8)\\
3: & $\;\;$ $\rho_{0,j}(\tau) \hspace{-0.025in} = \hspace{-0.025in} \left[ \frac{\bar{P}_{0,j}}{\lambda_0(\tau)} \hspace{-0.025in} - \hspace{-0.025in} \frac{W_j^{-}}{R_{0,j}} \right]^+$, 
   $\rho_{i,j}(\tau) \hspace{-0.025in} = \hspace{-0.025in} \left[ \frac{\bar{P}_{i,j}}{\lambda_i(\tau)} \hspace{-0.025in} - \hspace{-0.025in} \frac{W_j^{-}}{R_{i,j} G} \right]^+$ \hspace{-0.05in}; \\
4: & $\;\;$ IF $\left[ \bar{P}_{0,j} \log(W_j^{-}+\rho_{0,j}(\tau) R_{0,j})-\lambda_0(\tau)\rho_{0,j}(\tau) \right] >$ \\
   & $\;\;$ $\left[ \bar{P}_{i,j} \log(W_j^{-} + \rho_{i,j}(\tau) G R_{i,j}) - \lambda_i(\tau) \rho_{i,j}(\tau) \right]$ \\
5: & $\;\;\;\;\;$ Set $p_j(\tau)=1$ and $\rho_{i,j}(\tau)=0$; \\
6: & $\;\;$ ELSE  \\
7: & $\;\;\;\;\;$ Set $p_j(\tau)=0$ and $\rho_{0,j}(\tau)=0$; \\
8: & $\;\;$ END IF \\
9: & $\;\;$ MBS updates $\lambda_i(\tau+1)$ as in (\ref{eq:IterOptM0}) and (\ref{eq:IterOptM1}); \\
10: & $\;\;$ $\tau=\tau+1$; \\ 
11: & WHILE $\left( \sum_{i=0}^{N}(\lambda_i(\tau+1)-\lambda_i(\tau))^2 > \phi \right)$ \\
\hline 
\end{tabular}
\label{tab:OptDisM}
\end{center}
\end{table}

\subsubsection{Case of Multiple Interfering FBS's}

\paragraph{Formulation}

Finally, we consider the case of multiple interfering FBS's.  Assume that the coverages of some FBS's overlap with each other, as FBS 3 and 4 in Fig.~\ref{fig:netmod2}. They cannot use the same channel simultaneously, but have to compete for the available channels in the transmission phase.  Define {\em channel allocation variables} $c_{i,m}$ for time slot $t$ as: 
\begin{equation} \label{eq:cimt}
c_{i,m} = \left\{ \begin{array}{ll}
     1, & \mbox{if channel $m$ is allocated to FBS $i$} \\
     0, & \mbox{otherwise.}
                    \end{array} \right.                   
\end{equation}
Given an allocation, the expected number of available channels for FBS $i$ is $G_i \hspace{-0.025in} = \hspace{-0.025in} \sum_{m\in\mathcal{A}(t)} c_{i,m} P_m^A$.  

We use {\em interference graph} to model the case of overlapping coverages, which is defined below.

\begin{definition} 
An {\em interference graph} $G_I=(V_I,E_I)$ is an undirected graph where each vertex represents an FBS and each edge indicates interference between the two end FBS's.
\end{definition}

For the example given in Fig.~\ref{fig:netmod2}, we can derive an interference graph as shown in Fig.~\ref{fig:interferencegraph}. FBS 3 and 4 cannot use the same channel simultaneously, as summarized in the following lemma.  

\begin{figure}[!t]
\centering
\includegraphics[width=4.0in]{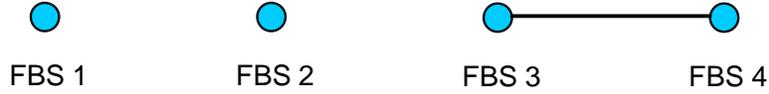}
\caption{Interference graph for the femtocell CR network shown in Fig.~\ref{fig:netmod2}.}
\label{fig:interferencegraph}
\end{figure}

\begin{lemma}
If channel $m$ is allocated to FBS $i$, the neighboring vertices of FBS $i$ in the interference graph $G_I$, denoted as $\mathcal{R}(i)$, cannot use the same channel $m$ simultaneously.  
\end{lemma}

Further define index variables $d_i^k$ as
\begin{equation} \label{eq:dik}
d_i^k = \left\{ \begin{array}{ll}
     1, & \mbox{if FBS $i$ is an endpoint of link $k \in G_I$} \\
     0, & \mbox{otherwise.}
                    \end{array} \right.                   
\end{equation}
The interference constraint can be described as $\sum_{i=1}^N d_i^k c_{i,m} \le 1$, for $m=0,\cdots,M$, and for all link $k \in G_I$. 
We then have the following problem formulation. 
\begin{eqnarray}\label{eq:ProbOptM}
 \mbox{maximize:} && \hspace{-0.25in} \sum_{j=1}^K  p_j \bar{P}_{0,j} \log(W_j^{-}+\rho_{0,j} R_{0,j}) + 
 \sum_{i=1}^N \sum_{j\in\mathcal{U}_i} q_j \bar{P}_{i,j} \log(W_j^{-} + \rho_{i,j} G_i R_{i,j}) \\ 
 \mbox{subject to:} && \hspace{-0.2in} \sum_{j=1}^K \rho_{0,j} \leq 1 \nonumber \\
 && \hspace{-0.25in} \sum_{j\in\mathcal{U}_i} \rho_{i,j} \le 1, \;\; i=1,\cdots,N \nonumber  \\
 && \hspace{-0.25in} p_j + q_j = 1, \;\; j=1,\cdots,K \nonumber \\
 && \hspace{-0.25in} G_i = \sum_{m\in\mathcal{A}(t)} c_{i,m} P_m^A, \;\; i=1,\cdots,N \nonumber \\
 && \hspace{-0.25in} \sum_{i=1}^N d_i^k c_{i,m} \le 1, m=0,\cdots,M, \mbox{for link }  k \in G_I, \nonumber \\
 && \hspace{-0.25in} \rho_{i,j}, p_j, q_j, c_{i,m} \ge 0, \;\; 
 i=1,\cdots,N,\; j=1,\cdots,K, \; m=0,\cdots,M. \nonumber
\end{eqnarray}

\paragraph{Solution Algorithm \label{subsubsec:ifbs}}

The optimal solution to problem (\ref{eq:ProbOptM}) depends on the channel allocation variables $c_{i,m}$. Problem (\ref{eq:ProbOptM}) can be solved with the algorithm in Table~\ref{tab:OptDisM} if the $c_{i,m}$'s are known.  Let $Q(\bm{c})$ be the suboptimal objective value for a given channel allocation $\bm{c}$, where $\bm{c}=[\vec{c}_1, \vec{c}_2, \cdots, \vec{c}_N]$ and $\vec{c}_i$ is a vector of elements $c_{i,m}$, for FBS $i$ and channels $m \in \mathcal{A}(t)$. If all the FBS's are disjointedly distributed with no overlap, each FBS can use all the available channels.  We have $c_{i,m}=1$ for all $i$ and $m \in \mathcal{A}(t)$, i.e., it is reduced to the case in Section~\ref{subsec:mulnifbs}.  

To solve problem (\ref{eq:ProbOptM}), we first apply a {\em greedy algorithm} to allocate the available channels in $\mathcal{A}(t)$ to the FBS's (i.e., to determine $\bm{c}$).  We then apply the algorithm in Table~\ref{tab:OptDisM} with the computed $\bm{c}$ to obtain a near-optimal solution.  Let $\bm{e}_{i,m}$ be a 
matrix with $1$ at position $\{i,m\}$ 
and $0$ at all other positions, representing the allocation of channel $m \in \mathcal{A}(t)$ to FBS $i$.  The greedy channel allocation algorithm is given in Table~\ref{tab:ChanAloc}, where the FBS-channel pair that can achieve the largest increase in 
$Q(\cdot)$ is chosen in each iteration. The worst case complexity of the greedy algorithm is $O(N^2 M^2)$.    

\begin{table}[!t]
\begin{center}
\caption{Channel Allocation Algorithm for Case of Interfering FBS's}
\begin{tabular}{ll}
\hline
1: & Initialize $\bm{c}$ to a zero matrix, FBS set $\mathcal{N}=\{1,\cdots,N\}$, and \\
   & FBS-channel set $\mathcal{C}=\mathcal{N}\times\mathcal{A}(t)$; \\
2: & WHILE ($\mathcal{C}$ is not empty) \\
3: & $\;\;$ Find FBS-channel pair $\{i',m'\}$, such that \\
   & $\;\;\;\;\;\;$ $\{i',m'\} = \stackbin[\{i,m\} \in \mathcal{C}]{}{\arg\max} \{Q(\bm{c} + \bm{e}_{i,m}) - Q(\bm{c}) \}$; \\
4: & $\;\;$ Set $\bm{c} = \bm{c} + \bm{e}_{i',m'}$; \\
5: & $\;\;$ Remove $\{i',m'\}$ from $\mathcal{C}$; \\
6: & $\;\;$ Remove $\mathcal{R}(i') \times m'$ from $\mathcal{C}$; \\
7: & END WHILE \\
\hline
\end{tabular}
\label{tab:ChanAloc}
\end{center}
\end{table}

\paragraph{Performance Lower Bound}

We next present a lower bound for the greedy algorithm. Let $e(l)$ be the $l$-th FBS-channel pair chosen in the greedy algorithm, and $\pi_l$ denote the sequence $\{e(1),e(2),\cdots,e(l)\}$.  The increase in object value (\ref{eq:ProbOptM}) due to the $l$-th allocated FBS-channel pair is denoted as
\begin{equation}
\Delta_l := \Delta(\pi_l, \pi_{l-1}) = Q(\pi_l) - Q(\pi_{l-1}). 
\end{equation}
Since $Q(\pi_0)=Q(\emptyset)=0$, we have 
\begin{eqnarray}
\sum_{l=1}^L \Delta_l &=& Q(\pi_L)-Q(\pi_{L-1})+ \cdots + Q(\pi_1) - Q(\pi_0) \nonumber \\
&=& Q(\pi_L) - Q(\pi_0) = Q(\pi_L). \nonumber
\end{eqnarray}

For two FBS-channel pairs $e(l)$ and $e(l')$, we say $e(l)$ {\em conflicts with} $e(l')$ when there is an edge connecting the FBS in $e(l)$ and the FBS in $e(l')$ in the interference graph $G_I$, and the two FBS's choose the same channel.  Let $\Omega$ be the global optimal solution. We define $\omega_l$ as the subset of $\Omega$ that conflicts with allocation $e(l)$ but not with the previous allocations $\{e(1), e(2), \cdots, e(l-1)\}$.  

\begin{lemma} \label{lemma1:5}
Assume the greedy algorithm in Table~\ref{tab:ChanAloc} stops in $L$ steps. The global optimal solution $\Omega$ can be partitioned into $L$ non-overlapping subsets $\omega_l$, $l=1,2,\cdots,L$. 
\end{lemma}

\begin{proof}
According to the definition of $\omega_l$, the $L$ subsets of the optimal solution $\Omega$ do not intersect with each other.  Assume the statement is false, then the union of these $L$ subsets is not equal to the optimal set $\Omega$.  Let the {\em set difference} be $\omega_{L+1} = \Omega \setminus (\cup_{l=1}^L \omega_l)$. By definition, $\omega_{L+1}$ does not conflict with the existing $L$ allocations $\{e(1), \cdots, e(L)\}$, meaning that the greedy algorithm can continue to at least the $(L+1)$-th step.  This conflicts with the assumption that the greedy algorithm stops in $L$ steps.  It follows that $\Omega = \cup_{l=1}^L \omega_l$. 
\end{proof}

Let $\Delta(\pi_2,\pi_1)=Q(\pi_2)-Q(\pi_1)$ denote the difference between two feasible allocations $\pi_1$ and $\pi_2$.  We next derive a lower bound on the performance of the greedy algorithm.  We assume two properties for function $\Delta(\pi_2,\pi_1)$ in the following. 

\begin{ppty}
Consider FBS-channel pair sets $\pi_1$, $\pi_2$, and $\sigma$, satisfying $\pi_1\subseteq\pi_2$ and $\sigma \cap \pi_2 =\emptyset$.  We have $\Delta(\pi_2\cup\sigma, \pi_1\cup\sigma) \le \Delta(\pi_2,\pi_1)$. 
\end{ppty}

\begin{ppty}
Consider FBS-channel pair sets $\pi$, $\sigma_1$, and $\sigma_2$ satisfying $\sigma_1 \cap \pi = \emptyset$, $\sigma_2 \cap \pi = \emptyset$, and $\sigma_1 \cap \sigma_2 = \emptyset$.  We have $\Delta(\sigma_1 \cup \sigma_2 \cup \pi, \pi) \le \Delta(\sigma_1 \cup \pi,\pi) + \Delta(\sigma_2 \cup \pi, \pi)$.  
\end{ppty}

In Property 1, we have $\sigma \cap \pi_1 =\emptyset$ since $\pi_1\subseteq\pi_2$ and $\sigma \cap \pi_2 =\emptyset$.  This property states that the incremental objective value does not get larger as more channels are allocated and as the objective value gets larger.  Property 2 states that the incremental objective value achieved by allocating multiple FBS-channel pair sets does not exceed the sum of the incremental objective values achieved by allocating each individual FBS-channel pair set.  These are generally true for many resource allocation problems~\cite{Kelly98}.

Since we choose the maximum incremental allocation in each step of the greedy algorithm, we have Lemma~\ref{lemma:step3} that directly follows Step 3 in Table~\ref{tab:ChanAloc}. 

\begin{lemma} \label{lemma:step3}
For any FBS-channel pair $\sigma \in \omega_l$, we have $Q(\pi_{l-1} \cup \sigma) - Q(\pi_{l-1}) = \Delta(\pi_{l-1} \cup \sigma, \pi_{l-1}) \le \Delta_l$.
\end{lemma}

\begin{lemma} \label{lemma:6}
Assume the greedy algorithm stops in $L$ steps, we have 
$$
  Q(\Omega)\le Q(\pi_L)+\sum_{l=1}^L \sum_{\sigma\in \omega_{l}} \Delta(\sigma \cup \pi_{l-1}, \pi_{l-1}).
$$
\end{lemma}

\begin{proof}
The following inequalities hold true according to the properties of the $\Delta(\cdot, \cdot)$ function: 
\begin{eqnarray}\label{eq:IneqProof}
Q((\cup_{i=l+1}^L \omega_i)\cup \pi_l) 
&=& Q((\cup_{i=l+2}^L \omega_i)\cup \pi_l) + 
\Delta((\cup_{i=l+1}^L \omega_i)\cup \pi_l,(\cup_{i=l+2}^L \omega_i)\cup \pi_l) \nonumber \\
&& \le Q((\cup_{i=l+2}^L \omega_i)\cup \pi_l)+\Delta(\omega_{l+1}\cup\pi_l,\pi_l) 
\nonumber \\
&&\le Q((\cup_{i=l+2}^L \omega_i)\cup \pi_{l+1})+\Delta(\omega_{l+1}\cup\pi_l,\pi_l) \nonumber \\
&& \le Q((\cup_{i=l+2}^L \omega_i)\cup \pi_{l+1}) + \sum_{\sigma\in\omega_{l+1}} \Delta(\sigma \cup \pi_l,\pi_l). \nonumber
\end{eqnarray}
We have $\pi_0=\emptyset$ and $\omega_{L+1}=\emptyset$ (see Lemma~\ref{lemma1:5}).  With induction from $l=0$ to $l=L-1$, 
we have $Q((\cup_{i=1}^L \omega_i) \cup \emptyset) = Q(\Omega)$ 
and $Q(\Omega) \le Q(\pi_{L}) + \sum_{l=1}^L \sum_{\sigma\in\omega_{l}} \Delta(\sigma\cup\pi_{l-1},\pi_{l-1})$.  
\end{proof}

\begin{lemma} \label{lemma:7}
The maximum size of $\omega_l$ is equal to the degree, in the interference graph $G_I$, of the FBS selected in the $l$-th step of the greedy algorithm, which is denoted as $D(l)$.
\end{lemma}

\begin{proof} 
Once FBS $i$ is allocated with channel $m$, the neighboring FBS's in $G_I$, $\mathcal{R}(i)$,  cannot use the same channel $m$ anymore due to the interference constraint. The maximum number of FBS-channel pairs that conflict with the selected FBS-channel pair $\{i,m\}$, i.e., the maximum size of $\omega_l$, is equal to the degree of FBS $i$ in $G_I$. 
\end{proof}

Then we have Theorem~\ref{th1:2} that provides a lower bound on the objective value achieved by the greedy algorithm given in Table~\ref{tab:ChanAloc}. 

\begin{theorem} \label{th1:2}
The greedy algorithm can achieve an objective value that is at least $\frac{1}{1+D_{max}}$ of the global optimum, where $D_{max}$ is the maximum node degree in the interference graph $G_I$ of the femtocell CR network. 
\end{theorem}

\begin{proof} 
According to Lemmas~\ref{lemma:6} and~\ref{lemma:7}, we have: 
\begin{eqnarray}\label{eq:OptBound}
&& \hspace{0.0in} Q(\Omega) \le Q(\pi_L) + \sum_{l=1}^L D(l)\Delta_l  
= Q(\pi_L)+\bar{D} \sum_{l=1}^L \Delta_l \nonumber \\
&& \hspace{0.375in} 
= (1 + \bar{D}) Q(\pi_L), 
\end{eqnarray}
where $\bar{D}=\sum_{l=1}^L D(l)\Delta_l / \sum_{l=1}^L \Delta_l$. 
The second equality is due to the facts that 
$\sum_{l=1}^L \Delta_l = Q(\pi_L)$.  

To further simplify the bound, we replace $D(l)$ with the maximum node degree $D_{max}$. 
We then have $\bar{D} \leq \sum_{l=1}^L D_{max} \Delta_l / \sum_{l=1}^L \Delta_l = D_{max}$ and  
\begin{equation} \label{eq:lowerbd}
  \frac{1}{1 + D_{max}} Q(\Omega) \le Q(\pi_L) \le Q(\Omega), 
\end{equation}
which provides a lower bound on the performance of the greedy algorithm. 
\end{proof}

When there is a single FBS in the CR network, we have $D_{max}=0$ and $Q(\pi_L) = Q(\Omega)$ according to Theorem~\ref{th1:2}. The proposed algorithm produces the optimal solution.  In the case of multiple non-interfering FBS's, we still have $D_{max}=0$ and can obtain the optimal solution using the proposed algorithm.  For the femtocell CR network given in Fig.~\ref{fig:netmod2} (with interference graph shown in Fig.~\ref{fig:interferencegraph}), we have $D_{max}=1$ and the low bound is a half of the global optimal. Note that (\ref{eq:OptBound}) provides a tighter bound for the optimum than (\ref{eq:lowerbd}), but with higher complexity. These are interesting performance bounds since they bound the achievable video quality, an application layer performance measure, rather than lower layer metrics (e.g., bandwidth or time share).

\subsection{Simulation Results \label{sec:sim4}}

We evaluate the performance of the proposed algorithms using MATLAB and JSVM 9.13 Video codec. Two scenarios are used in the simulations: a single FBS CR network and a CR network with interfering FBS's. In every simulation, we compare the proposed algorithms with the following three more straightforward heuristic schemes:
\begin{itemize}
       \item Heuristic 1 based on {\em equal allocation}: each CR user chooses the better channel (i.e., the common channel or a licensed channel) based on the channel conditions; time slots are equally allocated among active CR users;
       \item Heuristic 2 exploiting {\em multiuser diversity}: the MBS and each FBS chooses one active CR user with the best channel condition; the entire time slot is allocated to the selected CR user.
       \item {\em SCA-MAC} proposed in~\cite{Hsu07}: with this scheme, the successful transmission rate is evaluated based on channel packet loss rate and collision probability with primary users; the channel-user pair with the highest transmission probability is selected.
\end{itemize}
We choose SCA-MAC because it adopts similar models and assumptions as in this paper. Once the channels are selected, the same distributed algorithm is used for scheduling video data for all the three schemes.

We adopt the Raleigh block fading model and the packet loss probability is between [0.004, 0.028]. The frame rate is set to 30 fps and the GoP size is 16. The base layer mode is set to be AVC compatible. The motion search mode is set to Fast Search with search range 32. 
Each point in the figures presented in this section is the average of 10 simulation runs with different random seeds. We plot 95\% confidence intervals in the figures, which are generally negligible.

\subsubsection{Case of Single FBS}

In the first scenario, there are $M=8$ channels and the channel parameters $P_{01}^m$ and $P_{10}^m$ are set to 0.4 and 0.3, respectively, for all $m$.  The maximum allowable collision probability $\gamma_m$ is set to 0.2 for all $m$. There is one FBS and three active CR users.  Three Common Intermediate Format (CIF, 352$\times$288) video sequences are streamed to the CR users, i.e., {\em Bus} to CR user 1, {\em Mobile} to CR user 2, and {\em Harbor} to CR user 3. 
We have $T=10$ as the delivery deadline.  Both probabilities of false alarm $\epsilon$ and miss detection $\delta$ are set to 0.3 for all the FBS's and CR users, unless otherwise specified.

First we investigate the convergence of the distributed algorithm. The traces of the two dual variables are plotted in Fig.~\ref{fig:conv}. 
To improve the convergence speed, the correlation in adjacent time slots can be exploited. In particular, we set the optimal values for the optimization variables in the previous time slot as the initialization values for the variables in the current time slot. By doing so, the convergence speed can be improved. 
It can be seen that both dual variables converge to their optimal values after 300 iterations. After convergence, the optimal solution for the primary problem can be obtained.

\begin{figure}[!t]
\centering
\includegraphics[width=4.5in, height=3.0in]{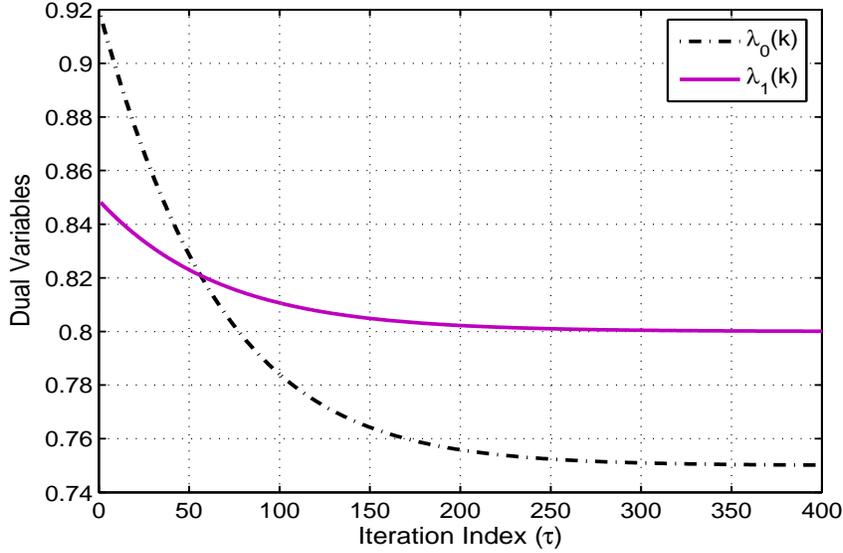}
\caption{Convergence of the two dual variables in the single FBS case.}
\label{fig:conv}
\end{figure}

Our proposed scheme achieves the best performance among the three algorithms, with up to 4.3 dB improvement over the two heuristic schemes and up to 2.5 dB over SCA-MAC. Such gains are significant with regard to video quality, since a 0.5 dB difference is distinguishable by human eyes. Compared to the two heuristic schemes and SCA-MAC, the video quality of our proposed scheme is well balanced among the three users, indicating better fairness performance.


In Fig.~\ref{fig:singleFBSmAll}, we examine the impact of the number of channels $M$ on received video quality. 
First, we validate the video quality measure used in our formulation by comparing the PSNR value computed using (\ref{eq:QuaMod}) with that computed from real decoded video frames.  
The average PSNR for three received videos are plotted in the figure. It can be seen that the real PSNRs are very close to those predicted by (\ref{eq:QuaMod}), with overlapping confidence intervals. This is also consistent with the results shown in Fig.~\ref{fig:mgs-rd}. 
Second, as expected, the more licensed channels, the more spectrum opportunities for CR users and the higher PSNR for received videos. 
SCA-MAC performs better than two heuristics, but is inferior to the proposed scheme.


\begin{figure}[!t]
\centering
\includegraphics[width=4.5in, height=3.0in]{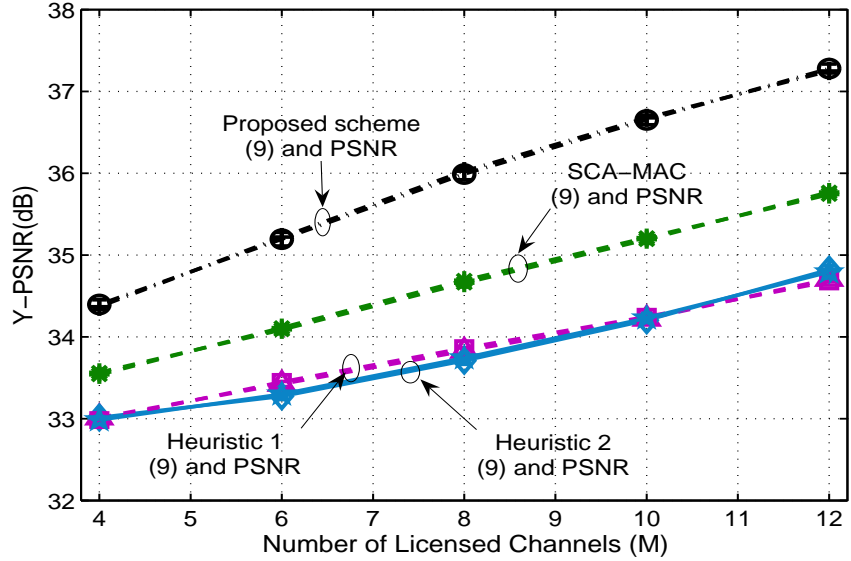}
\caption{Single FBS: received video quality vs. number of channels (computed with (9) and measured by PSNR).}
\label{fig:singleFBSmAll}
\end{figure}

We also plot the 
MS-SSIM of the received videos at the three CR users in Fig.~\ref{fig:singleFBSmSSIM}~\cite{Wang04}.  
Similar observations can be made from the MS-SSIM plot. All MS-SSIMs for the four curves are more than 0.97 and very close to 1. The proposed scheme still outperforms the other three schemes. In the remaining figures, we will use model predicted PSNR values, since the model (\ref{eq:QuaMod}) is sufficient to predict the real video quality. 

\begin{figure}[!t]
\centering
\includegraphics[width=4.5in, height=3.0in]{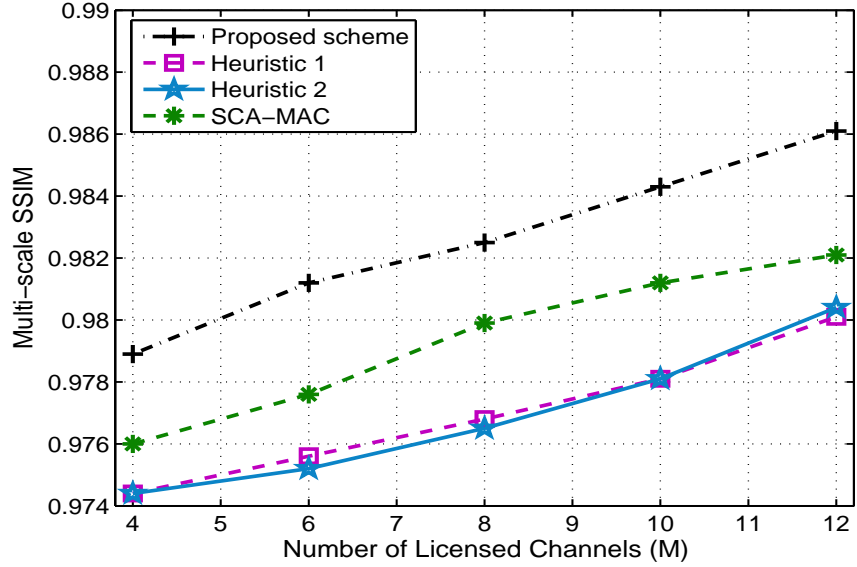}
\caption{Single FBS: received video quality vs. number of channels (measured by MS-SSIM).}
\label{fig:singleFBSmSSIM}
\end{figure}

In Fig.~\ref{fig:singleFBSeta}, we demonstrate the impact of channel utilization $\eta$ on received video quality. The average PSNRs achieved by the four schemes are plotted when $\eta$ is increased from 0.3 to 0.7. Intuitively, a smaller $\eta$ allows more spectrum opportunities for video transmission. This is illustrated in the figure where all the three curves decrease as $\eta$ gets larger. The performance of both heuristics are close and the proposed scheme achieves a gain about 3 dB over the heuristics and 2 dB over SCA-MAC.

\begin{figure}[!t]
\centering
\includegraphics[width=4.5in, height=3.0in]{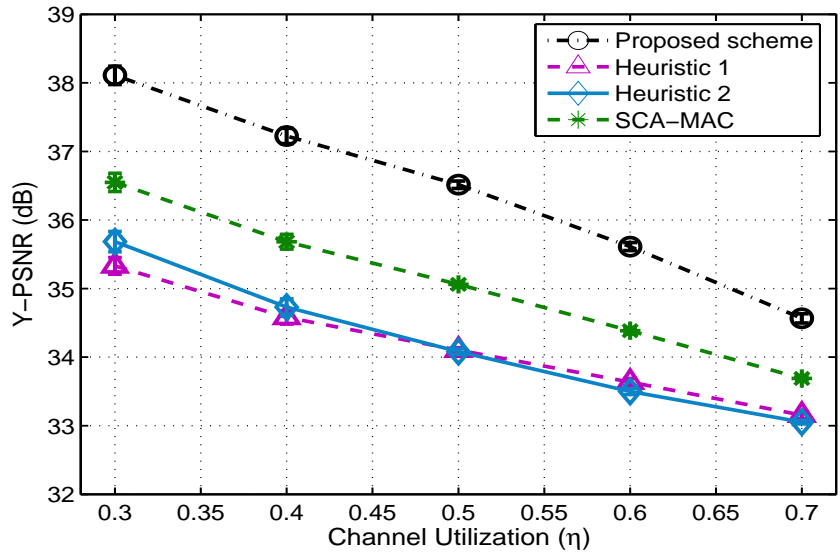}
\caption{Single FBS: received video quality vs. channel utilization.}
\label{fig:singleFBSeta}
\end{figure}

We also compare the MGS and FGS videos while keeping other parameters identical. We find that MGS video achieves over 0.5 dB gain in video quality over FGS video. The results are omitted for brevity. 




\subsubsection{Case of Interfering FBS's}

We next investigate the second scenario with three FBS's, and each FBS has three active CR users. Each FBS streams three different videos to the corresponding CR users. The coverages of FBS 1 and 2 overlap with each other, and the coverages of FBS 2 and 3 overlap with each other. 

In Fig.~\ref{fig:multiFBSm}, we examine the impact of the number of channels $M$ on the received video quality. The average PSNRs of all the active CR users are plotted in the figure when we increase $M$ from 12 to 20 with step size 2. As mentioned before, more channels imply more transmission opportunities for video transmission. In this scenario, heuristic 2 (with a multiuser diversity approach) outperforms heuristic 1 (with an equal allocation approach).  But its PSNRs are still about 0.3 $\sim$ 0.5 dB lower that those of the proposed algorithm. The proposed scheme has up to 0.4 dB improvement over SCA-MAC. In Fig.~\ref{fig:multiFBSm}, we also plot an upper bound on the optimal objective value, which is obtained as in (\ref{eq:OptBound}). It can be seen that the performance of our proposed scheme is close to optimal solution since the gap between the upper bound and our scheme is generally small (about 0.5 dB).

\begin{figure}[!t]
\centering
\includegraphics[width=4.5in, height=3.0in]{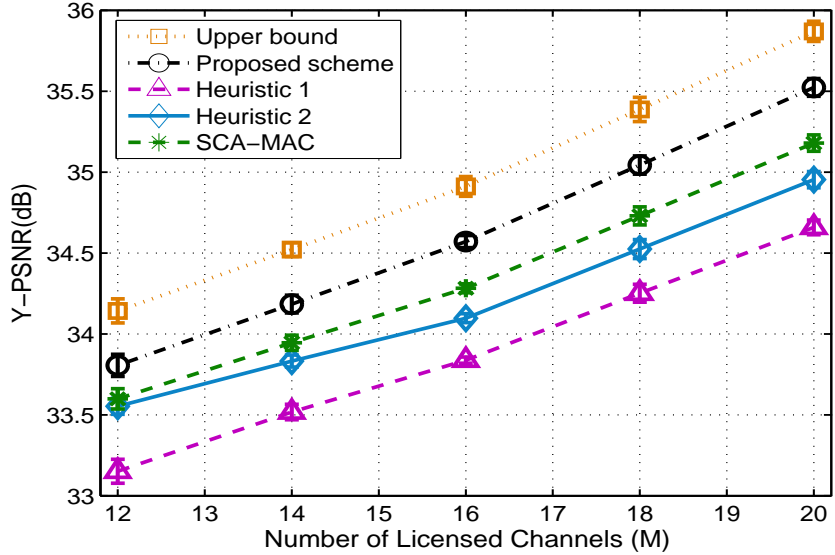}
\caption{Interfering FBS's: received video quality vs. number of channels.}
\label{fig:multiFBSm}
\end{figure}



Next, we examine the impact of sensing errors on the received video quality.  In Fig.~\ref{fig:multiFBSepsilon}, we test five pairs of $\{\epsilon, \delta\}$ values: \{0.2,0.48\}, \{0.24,0.38\}, \{0.3,0.3\}, \{0.38,0.24\}, and \{0.48,0.2\}. 
It is interesting to see that the performance of all the four schemes get worse when the probability of one of the two sensing errors gets large.  We can trade-off between false alarm and miss detection probabilities to find the optimal operating point for the spectrum sensors. 
Moreover, the dynamic range of video quality is not big for the range of sensing errors simulated, compared to that in Fig.~\ref{fig:multiFBSm}. This is because both sensing errors are modeled and treated in the algorithms. Again, our proposed scheme outperforms the two heuristic schemes and SCA-MAC with considerable margins for the entire range.

\begin{figure}[!t]
\centering
\includegraphics[width=4.5in, height=3.0in]{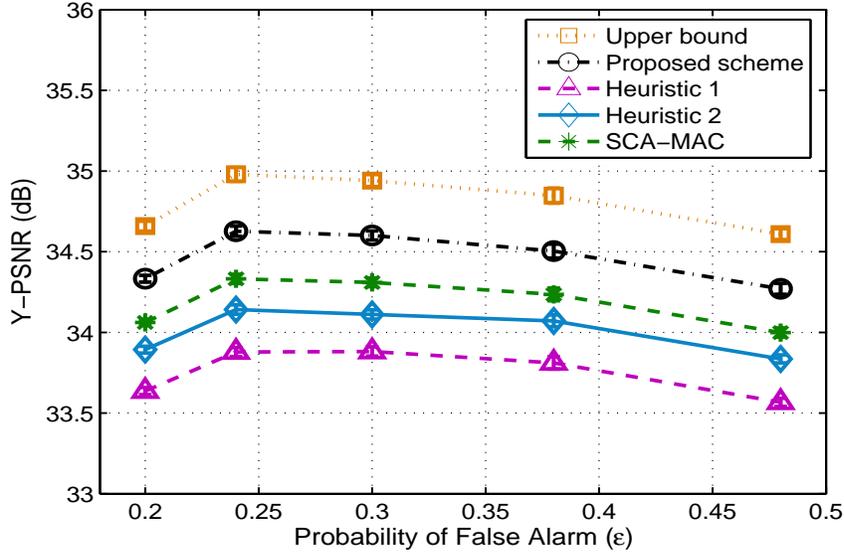}
\caption{Interfering FBS's: received video quality vs. sensing error probability.}
\label{fig:multiFBSepsilon}
\end{figure}

We also investigate the impact of the bandwidth of the common channel $B_0$. In this simulation, we fix $B_1$ at 0.3 Mbps and increase $B_0$ from 0.1 Mbps to 0.5 Mbps with step size 0.1 Mbps. The results are presented in Fig.~\ref{fig:multiFBSCtrlBW}. We notice that the average video quality increases rapidly as the common channel bandwidth is increased from 0.1 Mbps to 0.3 Mbps. Beyond 0.3 Mbps, the increases of the PSNR curves slow down and the curves get flat. This implies that a very large bandwidth for the common channel is not necessary, since the gain for additional bandwidth diminishes as $B_0$ gets large. Again, the proposed scheme outperforms the other three schemes and the gap between our scheme and the upper bound is small.

\begin{figure}[!t]
\centering
\includegraphics[width=4.5in, height=3.0in]{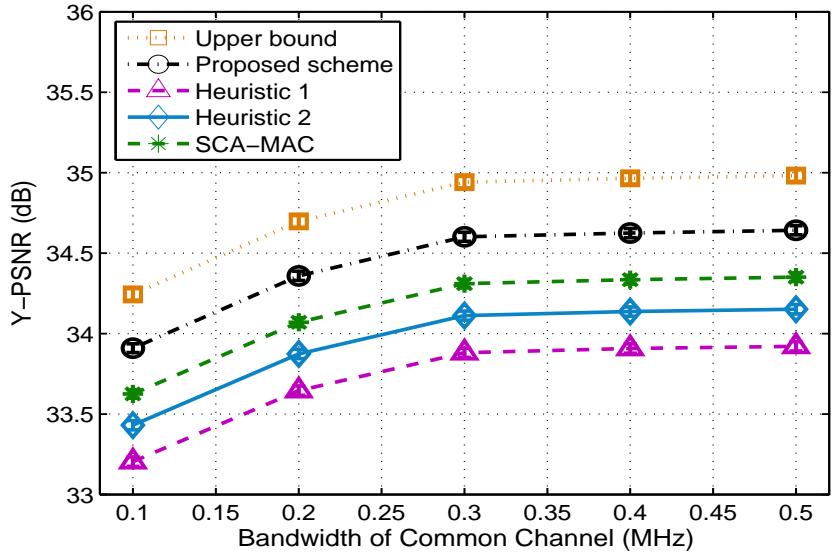}
\caption{Interfering FBS's: received video quality vs. bandwidth of the common channel.}
\label{fig:multiFBSCtrlBW}
\end{figure}

Next, we stop the distributed algorithm after a fixed amount of time, 
and evaluate the suboptimal solutions. In particular, we vary the duration of time slots, and let the distributed algorithm run for 5\% of the time slot duration at the beginning of the time slot. Then the solution obtained this way will be used for the video data transmissions. The results are presented in Fig.~\ref{fig:timedura}. 
It can be seen that when the time slot is 5 ms, the algorithm does not converge after 5\%$\times$5 = 0.25 ms and the PSNR produced by the distributed algorithm is close to that of Heuristic 1, and lower than those of Heuristic 2 and SCA-MAC. When the time slot is sufficiently large, the algorithm can get closer to the optimal and the proposed algorithm produces better video quality as compared to the two heuristic algorithms and SCA-MAC. Beyond 20 ms, the increase in PSNR is small since all the curves gets flat. Therefore the proposed algorithm could be useful even when there is no time for it to fully converge to the optimal. 

\begin{figure}[!t]
\centering
\includegraphics[width=4.5in, height=3.0in]{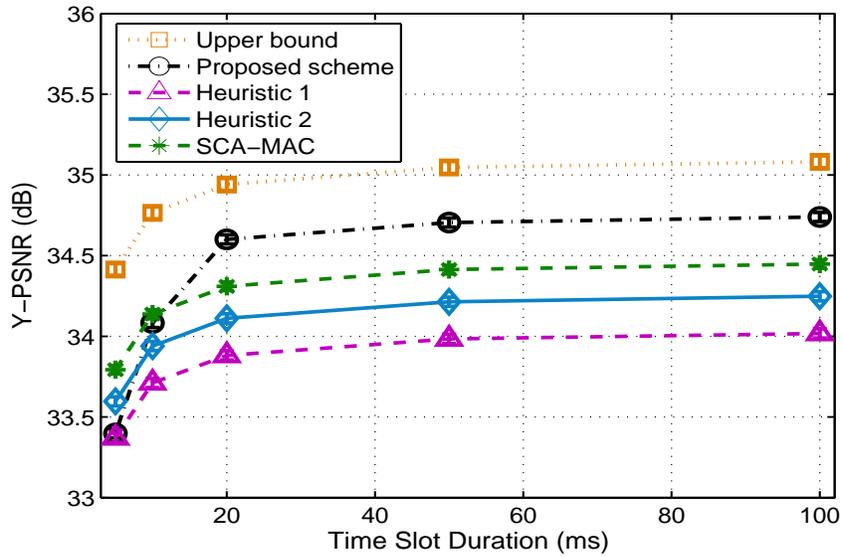}
\caption{Video quality achieved by the algorithms when they are only executed for 5\% of the time slot duration.}
\label{fig:timedura}
\end{figure}



During the simulations, we find the collision rate with primary users are strictly kept below the prescribed collision tolerance $\gamma$. These results are omitted for brevity. 
 

\section{Conclusions}\label{sec:femto_conc}
In this paper, we first investigated data multicast in femtocell networks consisting of an MBS and multiple FBS's. We adopted SC and SIC for multicast data and investigated how to assign transmit powers for the packet levels. The objective was to minimize the total BS power consumption, while guaranteeing successful decoding of the multicast data at each user.  We developed optimal and near-optimal algorithms with low computational complexity, as well as performance bounds. The algorithms were evaluated with simulations and are shown to outperform a heuristic with considerable gains. 

Next, we investigated the problem of streaming multiple MGS videos in a femtocell CR network. We formulated a multistage stochastic programming problem considering various design factors across multiple layers.  We developed a distributed algorithm that can produce optimal solutions in the case of non-interfering FBS's, and a greedy algorithm for near-optimal solutions in the case of interfering FBS's with a proved lower bound. 
The proposed algorithms are evaluated with simulations and are shown to outperform three alternative schemes with considerable gains. 

\bibliographystyle{IEEEtran}
\bibliography{cr_video_femto,MyWork}

\end{document}